\documentclass[12pt, draftclsnofoot, onecolumn]{IEEEtran}

\usepackage{setspace}

\usepackage{subcaption}

\usepackage{multicol}
\usepackage{amsmath, amssymb, bm, cite, epsfig, psfrag, mathtools}
\usepackage{epstopdf}
\usepackage{changepage}
\usepackage{graphicx}
\usepackage{fancyhdr}
\usepackage{bbm}
\usepackage{algorithm,algpseudocode}
\usepackage{array}
\usepackage[margin=10pt,font=small]{caption}
\usepackage{bbm}
\usepackage{multirow}
\usepackage[usenames,dvipsnames]{xcolor}

\usepackage{etoolbox}
\usepackage{pbox}
\usepackage[width=0.48\textwidth]{caption}
\usepackage{colortbl}
\usepackage{bm}
\usepackage{caption}
\usepackage{amsmath}
\usepackage{color}
\usepackage{enumitem}
\usepackage{datetime}
\usepackage{amsthm}
\usepackage{dsfont}
\usepackage{mathtools}
\usepackage[T1]{fontenc}

\newtoggle{conference}
\togglefalse{conference} % Use for arxiv version -- also change documentclass
\newcommand{\eqdef}{\stackrel{\Delta}{=}}

\def\beq{\begin{equation}}
\def\eeq{\end{equation}}
\def\beqa{\begin{eqnarray}}
\def\eeqa{\end{eqnarray}}
\def\beqan{\begin{eqnarray*}}
\def\eeqan{\end{eqnarray*}}

\def\argmin{\mathop{\mathrm{arg\,min}}}
\def\argmax{\mathop{\mathrm{arg\,max}}}

\newtheorem{remark}{Remark}

\newtheorem{definition}{Definition}
\newtheorem{theorem}{Theorem}

\newtheorem{corollary}{Corollary}

\def\tm1{t\! - \! 1}
\def\tp1{t\! + \! 1}

\def\dbf{\mathbf{d}}

\def\fbf{\mathbf{f}}

\def\pbf{\mathbf{p}}

\def\qbf{\mathbf{q}}

\def\xbf{\mathbf{x}}

\def\Cbf{\mathbf{C}}

\def\Qbf{\mathbf{Q}}

\newcommand{\Vc}{{\cal V}}

\def\Xi{{\text{Eff}} }

\def\factor{\tau}
\def\length{\Omega}
\def\cache{\mu}

\def\order{\chi}

\def\ach{ { ach } }
\usepackage{units}

\begin{document}

% **************************************** TITLE PAGE **************************************************
\title{{Rate-Distortion-Memory Trade-offs in Heterogeneous Caching Networks}}

\author{\normalsize{
Parisa Hassanzadeh, Antonia M. Tulino, Jaime Llorca, Elza Erkip
   %  Eldon Tyrell\IEEEauthorrefmark{4}
   	\thanks{This work has been supported in part by NSF under grant \#1619129 and in part by NYU WIRELESS Industrial Affiliates Program. This paper was presented in part at Allerton 2015 \cite{hassanzadeh2015distortion}.}}
%\thanks{P. Hassanzadeh  and  E. Erkip are with the ECE Department of New York University, Brooklyn, NY. Email: \{ph990, elza\}@nyu.edu}
%\thanks{J. Llorca  and A. Tulino are with Bell Labs, Nokia, Holmdel, NJ, USA. Email:  \{jaime.llorca, a.tulino\}@nokia.com}
%\thanks{A. Tulino is with the DIETI, University of Naples Federico II, Italy. Email:  \{antoniamaria.tulino\}@unina.it}
}

\maketitle

% ********************************************************************************************************
% ABSTRACT
% ********************************************************************************************************
\vspace{-0.75in}
\begin{abstract}
%Mobile network operators are considering 
Caching at the wireless edge can be used to keep up with the increasing demand for high-definition wireless video streaming. By prefetching popular content into memory at wireless access points or end-user devices, requests can be served locally, relieving strain on expensive backhaul. In addition, using network coding allows the simultaneous serving of distinct cache misses via common coded multicast transmissions,  resulting in significantly larger load reductions compared to those achieved with traditional delivery schemes. %Most prior works do not exploit the properties of video and simply treat content as fixed-size files that users would like to fully download. 
Most prior works  simply treat video content as fixed-size files that users would like to fully download. This work is motivated by the fact that video can be coded in a scalable fashion and that the decoded video quality depends on the number of layers a user receives in sequence. Using a Gaussian source model, caching and coded delivery methods  
are designed to minimize the squared error distortion at end-user devices in a rate-limited caching network.  The framework is very general and accounts for heterogeneous cache sizes, video popularities and user-file play-back qualities. As part of the solution, a new decentralized scheme for lossy cache-aided delivery subject to preset user distortion targets is proposed, which further generalizes prior literature to a setting with file heterogeneity.  
 \end{abstract}

\begin{IEEEkeywords} %Caching networks, coded multicast, scalable coding, successive refinement, distributed lossy source coding
Caching networks, coded multicast, scalable coding, successive refinement,  lossy source coding 	
\end{IEEEkeywords}

%Caching is considered as one of the strategies to keep up with the increasing demand for high-definition wireless video streaming. By prefetching popular content into memory at wireless access points or end-user devices, requests can be served locally, relieving strain on expensive backhaul. In addition, using network coding allows the simultaneous serving of distinct cache misses via common coded multicast transmissions,  resulting in significantly larger load reductions compared to those achieved with traditional delivery schemes. Most prior works simply treat video content as fixed-size files that users would like to fully download. Our work is motivated by the fact that video can be coded in a scalable fashion and that the decoded video quality depends on the number of layers a user receives in sequence. Using a Gaussian source model, caching and coded delivery methods are designed to minimize the squared error distortion at end-user devices in a rate-limited caching network.  Our framework is very general and accounts for heterogeneous cache sizes, video popularities and user-file play-back qualities. As part of our solution, a new decentralized scheme for lossy cache-aided delivery subject to preset user distortion targets is proposed, which further generalizes prior literature to a setting with file heterogeneity.  

% ********************************************************************************************************
% INTRODUCTION
% ********************************************************************************************************
 
\section{Introduction}~\label{sec: Introduction}
With the recent explosive growth in cellular video traffic, wireless operators are heavily investing in making infrastructural improvements such as increasing base station density and offloading traffic to Wi-Fi.
Caching is a technique to reduce traffic load by exploiting the high degree of asynchronous content reuse and the fact that storage is cheap and ubiquitous in today's wireless devices~\cite{molisch14caching}. During off-peak periods when network resources are abundant, popular content can be stored at the wireless edge,  %(e.g., access points or end user devices),
 so that peak hour demands can be met with reduced access latencies and bandwidth requirements.
 
The simplest form of caching is to store the most popular video files at every edge cache~\cite{wang14cache}. Requests for popular cached files can then be served locally, while cache misses need to be served by the base station, achieving what is referred to as a local caching gain. However,  replicating the same content on many devices can result in an inefficient use of the aggregate cache capacity~\cite{golrezaei12femtocaching}.
In fact, recent studies~\cite{maddah14fundamental, maddah14decentralized,  ji15order,yu2017characterizing} have shown that making users store different portions of the video files creates coded multicast opportunities that enable a global caching gain. In \cite{yu2017characterizing}, the memory-rate trade-off for the worst-case and average demand is characterized within a factor of two of an information theoretic lower bound for uniformly popular files. Caching networks have been extended to various settings including setting with random demands \cite{niesen2017coded,ji15order}, online caching \cite{pedarsani2016online}, noisy channels \cite{bidokhti2018noisy}, and correlated content \cite{ITjournal,JSAC2018}. A comprehensive review of existing work on caching networks can be found in \cite{paschos2018role}.

%The case of random demands according to a arbitrary popularity distribution is analyzed in \cite{niesen2017coded,ji15order,zhang2018coded}, where the authors propose schemes that are able to achieve performance within a constant factor of the optimal performance.  The caching network has been studied for various settings including multiple per-user demands \cite{ji2015caching}, online caching \cite{pedarsani2016online}, in interference channels \cite{maddah2019cache,naderializadeh2017fundamental}, noisy channels \cite{bidokhti2018noisy,amiri2018caching}, and settings with correlated content \cite{hassanzadeh2017rate,ITjournal,JSAC2018}.  A comprehensive review of existing work on caching networks can be found in \cite{paschos2018role}.

While existing work on wireless caching is motivated by video applications, the majority do not exploit specific properties of video in the caching and delivery phases. The cache-aided delivery schemes available in literature are based on fixed-to-variable source encoding, designed to minimize the aggregate rate on the shared link so that the requested files are recovered  in a lossless manner \cite{maddah14fundamental,  ji15order,yu2017characterizing,maddah14decentralized}. However, all video coders allow for lossy recovery \cite{wang2002video}. In particular, in scalable video coding (SVC)~\cite{SVC}, video files are encoded into layers such that the base layer contains the lowest quality level and additional enhancement layers allow successive improvement of the video streaming quality. SVC strategies are especially suitable for heterogeneous wired and wireless networks, since they encode video into a scalable bitstream such that video reconstructions of different spatial and temporal resolutions, and hence different qualities, can be generated by simply truncating the scalable bitstream. This scalability accommodates network requirements such as bandwidth limitations, user device capability, and quality-of-service restrictions in video streaming applications \cite{sun2007overview}.

In this work, we consider a lossy cache-aided network where the caches are used to enhance  video  reconstruction quality at user devices. We consider a scenario in which users store compressed files at different encoding rates (e.g., video layers in SVC). Upon delivery of requests, depending on the available network resources, users receive additional layers that successively refine the reconstruction quality. By exploiting scalable compression, we investigate the fundamental limits in caching networks with throughput limitations. We allow users to have different preferences in reconstruction quality for each library file, and assume that files have possibly different distortion-rate functions. These assumptions further account for the diversity of multimedia applications being consumed in wireless networks (e.g., YouTube videos vs 3D videos or augmented reality applications), and with respect to requesting users' device capabilities (e.g., 4K vs 1080p resolution).  Our goal is to design caching  schemes that, for a given broadcast rate, minimize the average distortion experienced at user devices.

\subsection{Related Work}
%We study a lossy caching network where receivers are connected to the sender through a shared rate-limited link, and the caching and delivery strategies are designed to minimize the expected distortion across the network  for a given broadcast rate. 
As discussed above, most literature on caching considers lossless recovery of files with the goal of minimizing the total rate transmitted over the shared link, in order to recover all requested fixed-size  video files in whole \cite{maddah14fundamental, maddah14decentralized, ji14average, ji15order,yu2017characterizing}.   There are only a few works that study the lossy cache-aided broadcast network \cite{timo2016rate,yang2018coded,ibrahim2018coded}.  In \cite{timo2016rate} the authors study the delivery rate, cache capacity and reconstruction
distortion trade-offs in a network with arbitrarily correlated sources for the single-user network and some special cases of a two-user problem. Similarly to  this paper,  \cite{yang2018coded}  and \cite{ibrahim2018coded}  assume successively refinable sources in a setting where receivers have heterogeneous distortion requirements. In \cite{yang2018coded}, the authors study the problem of minimizing the worst-case delivery rate for Gaussian sources and heterogeneous distortion requirement at the users. They characterize the optimal delivery rate for the two-file two-user case, and propose efficient centralized and decentralized caching schemes based on successive refinement coding for the general case.  The work in \cite{ibrahim2018coded}  extends  \cite{yang2018coded} to a setting where the server not only designs the users' cache contents, but also optimizes their cache sizes subject to a total memory budget.

\subsection{Contributions}
Our work differs from \cite{timo2016rate,yang2018coded,ibrahim2018coded} in a number of ways. Compared to \cite{timo2016rate} which considers a single-cache network, we have a large network with arbitrary number of receivers, each equipped with a cache memory of different capacity. The works in \cite{yang2018coded,ibrahim2018coded} minimize the worst-case rate transmitted over the broadcast link for a set of predetermined reconstruction distortion requirements at each user, while we minimize the expected distortion across the network subject to a given broadcast rate for a more general setting as elaborated below. Our main contributions are summarized as follows:

\begin{enumerate}
 \item We formulate the problem of efficient lossy delivery of sources over a heterogeneous rate-limited broadcast caching network via information-theoretic tools, and study the trade-off between user cache sizes, broadcast rate and the expected reconstruction distortion across users and demands. We allow for sources to have different distortion-rate functions, %(variances in the case of Gaussian sources),  
  and for users to have different cache sizes and different  demand distributions. %As mentioned earlier, the works in \cite{yang2018coded,ibrahim2018coded} consider fixed reconstruction distortion requirements at each user for all files in the library, and minimize the delivery rate. As a means to studying the rate-distortion-memory trade-off in our setting, we also provide a solution to a generalized version of the problem considered in \cite{yang2018coded}, which is further explained in Contribution 5.

\item We propose a class of cache-aided delivery schemes, in which, to limit the computational complexity and reduce the communication overhead, the sender only takes into account users' local cached content during the delivery phase, and generates the transmit message independently for each receiver without exploring multicast coding opportunities.   We refer to this scheme, presented in Sec.~\ref{sec: Unicast}, as the \textit{Local Cache-aided Unicast (LC-U)} scheme.  We show that the optimal caching policy in LC-U  admits a reverse water-filling type solution, which can be implemented locally and independently across users, without the need of global coordination. %The content cached at each receiver provides a low quality representation of each video stream. Then, for each file request the sender  effectively computes the optimal delivery rates allocated to each user just based on the quality level available at the corresponding local cache.

\item We propose another class of   schemes in Sec.~\ref{sec: Multicast} referred to as the \textit{Cooperative  Cache-aided  Coded Multicast (CC-CM)} scheme. In CC-CM,  the sender designs the caching and delivery phases jointly across all receivers based on global network knowledge (user cache contents and demand distributions, and file rate-distortion functions), and compresses the files accordingly. In this scheme global network knowledge is used to fill user caches and to construct codes that fully exploit the multicast nature of a wireless system.  %In Sec.~\ref{sec:Simulations}, we  numerically show that CC-CM  offers notable performance improvements over LC-U in terms of average file reconstruction distortion. %In this case, the optimal policy requires joint optimization of the compression rates at which files are cached by each user. After users place their requests, the sender, with knowledge of the users' cache contents, computes a common multicast codeword that simultaneously delivers additional enhancement layers to each user.

\item In Sec. \ref{sec:RAP}, we present a coded delivery scheme that can be adopted by CC-CM to implement the caching phase and to deliver a portion of the multicast message. We refer to this scheme, which is a generalization of the scheme proposed in \cite{ji15order} to a setting with heterogeneous cache sizes, demand distributions, and where users are interested in receiving possibly {degraded versions (different-length portions) of a given file  in the library}, as the {\em Random Fractional caching with Greedy Constrained Coloring (RF-GCC)}. We provide upper bounds on the per-demand and average delivery rates achieved with  RF-GCC. We note that RF-GCC allows the generalization of the problem studied in \cite{yang2018coded} where (i) files have different distortion-rate functions, and (ii) users have different reconstruction distortion targets for each library file. When specialized to the setting in \cite{yang2018coded}, our results  show that RF-GCC achieves equal or better worst-case delivery rate compared to the decentralized scheme proposed in \cite{yang2018coded}. %{\RED We don't really have a plot for this in our paper}

\item In Sec. \ref{sec:RAP-GCCOptimization}, we describe how  RF-GCC  presented in Sec. \ref{sec:RAP} can be used to deliver part of the transmitted message in  CC-CM, introduced in Sec. \ref{sec: Multicast}. The remaining part is delivered via unicast, and based on these two components we characterize the rate-distortion-memory trade-off achieved with CC-CM. In Sec.~\ref{sec:Simulations}, we  numerically show that CC-CM  offers notable performance improvements over LC-U in terms of average file reconstruction distortion.

% \item \BLUE As a means to designing the CC-CM scheme,  we solve a generalization of the problem studied in \cite{yang2018coded} to (i) files having different distortion-rate functions, and (ii) users having different reconstruction distortion targets for each library file. In Sec.~\ref{sec:RAP}, we propose an achievable scheme based on a generalization of the scheme proposed in \cite{ji15order} to a setting with heterogeneous cache sizes, demand distributions, and where users are interested in receiving possibly {degraded versions (different-length portions) of a given file  in the library}. For this scheme, we derive an upper bound on the rate-memory trade-off for any given user demand combination as well as for the expected rate-memory trade-off, and use the results to solve the main problem of interest in this paper.  When specialized to the setting in \cite{yang2018coded}, our results  show that our proposed scheme achieves equal or better worst-case delivery rate compared to the decentralized scheme proposed in \cite{yang2018coded}. {\RED We don't really have a plot for this in our paper}

%\item \BLUE In Sec.~\ref{sec:RAP-GCCOptimization}, we discuss how the RF-GCC proposed in Sec.~\ref{sec:RAP} can be adopted by the CC-CM scheme to fill receiver caches and to deliver the requested content through coded multicast transmissions, based on which we characterize the rate-distortion-memory trade-off achieved with CC-CM. 

\end{enumerate}

%\subsection{Paper Organization}
%This paper is organized as follows. The system model is presented in Sec. \ref{sec: ProblemSetting}. Sec. \ref{sec: Unicast} introduces LC-U, an achievable scheme based on  local caching and unicast transmissions, and Sec. \ref{sec: Multicast} proposes CC-CM, a scheme based on cooperative caching and coded multicast transmissions, and its performance is analyzed in Sec.~\ref{sec:RAP-GCCOptimization}. Finally, Sec. \ref{sec:Simulations} presents numerical results that illustrate the achievable distortion-memory trade-offs discussed in the paper.

%{\em Notation:} The set of integers $\{1,\dots,n\}$ is denoted by $[n]$.  We use $\{x_{i,j} \}$ in short to denote the entire set of elements $x_{i,j}$, for all $i\in\mathcal A$ and all $j\in\mathcal B$, and $(x)^+$ is used to denote $\max\{x,0 \}$. We use $i\ni \xbf$ to indicate that $i$ is one of the elements of vector $\xbf$.

\section{System Model and Problem Statement}~\label{sec: ProblemSetting}
\vspace{-1cm}
\subsection{Source Model}\label{subsec:source}
Consider a library composed of $N$ independent files indexed by $\{1,\dots,N \} \triangleq[N]$ and generated by an $N$-component  memoryless  source (N-MS) over finite alphabets $\mathcal W_1,\dots,\mathcal W_N$ with a  pmf  $p(w_1,\dots,w_N)=$ $ p(w_1),\dots, p(w_N)$. For a block length $F$, file $n\in[N]$ %\footnote{$[n]$ denotes the  discrete set of integers	from $1$ to $n$, i.e., $[n]\triangleq\{1,\dots,n \}$.} 
is represented by a sequence $W_n^F = (W_{n1},\dots, W_{nF})$, %, referred to as a file $n$ hereafter, 
where  $W_n^F \in {\mathcal W}_n^F $. For  a given reconstruction alphabet $\widehat{\mathcal W}_n$, an estimate of file $W_n^F$, $n\in[N]$, is represented by  ${\widehat W}_n^F\in\widehat{\mathcal W}_n^F$, and the distortion between the file and its reconstruction is measured  by a single letter distortion function $D_n: \mathcal W_n \times \widehat{\mathcal W}_n \rightarrow {\mathbb R}^{+}$, as $D_n( {W}_{n}^F,\widehat{W}_{n}^F) = \frac{1}{F} \sum\limits_{i=1}^F D_n( {W}_{n,i},\widehat{W}_{n,i})$.

%\begin{align}
%D_n( {W}_{n}^F,\widehat{W}_{n}^F) = {1}/{F} \sum\limits_{i=1}^F D_n( {W}_{n,i},\widehat{W}_{n,i}) . \label{eq: file dist}
%\end{align}
 We consider successively refinable sources, as defined in \cite{SuccessiveRefin}, where each source can be compressed in multiple stages such that the optimal distortion is achieved at each stage without incurring rate loss relative to its single-description representation.   Specifically, in the case of two stages, consider a first description of the file $W_n^F$  compressed  at rate $R^{(1)}$ bits/source-sample incurring distortion $D^{(1)}$, and an additional description that is compressed at rate $R^{(2)}-R^{(1)}$ bits/source-sample, such that the reconstruction resulting from the two-stage description has distortion $D^{(2)} \leq D^{(1)}$. Then,   the underlying N-MS is successively refinable if it is possible to construct codes such that $D^{(1)} = D(R^{(1)})$ and $D^{(2)} = D(R^{(2)})$, where $D(R)$  denotes the source  distortion-rate function. This suggests that the descriptions at each stage are optimal and the distortion-rate limit at both stages can be simultaneously achieved. 

Without loss of generality,  %for ease of exposure and analytical tractability, 
and based on the fact that Gaussian sources with squared error distortion are successively refinable, we assume that the source distribution is Gaussian %source $W_n^F$ consists of $F$ i.i.d. Gaussian samples 
with variance $\sigma_{n}^{2}$ and distortion-rate function $D_{n}(r) = \sigma_{n}^{2} 2^{-2r}$\cite{Cover}. Note that the compression setting considered in this paper is applicable to a video streaming application, in which each file represents a video segment compressed using SVC \cite{SVC}. In SVC, the single-stream video is encoded into multiple components, referred to as {\em layers}, such that the scalable video content is a combination of one {\em base layer} and multiple additional {\em enhancement layers}.  The base layer contains the lowest spatial, temporal and quality representation of the video, while  enhancement  layers can improve the quality of the video file reconstructed at the receiver. Note that an enhancement layer is useless unless the receiver has access to the base layer and all preceding enhancement layers. The reconstructed video quality (distortion) in SVC depends on the total number of layers received in sequence. 

\subsection{Cache-Aided Content Distribution Model}
Consider a cache-aided broadcast system, where one sender (e.g., base station) is connected through an error-free  rate-limited shared link to $K$ receivers (e.g., access points or user devices).  
The sender has access to a content library generated by an $N$-MS source as described in Sec.~\ref{subsec:source}. Receiver $k\in[K]$ has a cache of size $M_{k}$ bits/sample, or equivalently, $M_{k}F$ bits, as shown in  Fig.~\ref{fig: System}. Receiver $k\in[K]$ requests files from the library independently according to demand distribution $\qbf_k = (q_{k,1},\dots,q_{k,N})$, assumed to be known at the sender, where $q_{k,n} \in [0,1]$ for all $n\in[N]$, $\sum_{n=1}^N q_{k,n} = 1$, and  $q_{k,n}$ denotes the probability that  receiver  $k$ requests  file  $n$. 

 The cache-aided content distribution system operates in two phases: 
 \begin{itemize}
	\item[(i)]{{\em Caching Phase:}}
	%The caching 
	This phase occurs during a period of low network traffic. In this phase, all receivers have access to the entire library for filling their caches. Designing the cache content can be done locally by the receivers based on their local information, or globally in a cooperative manner either directly by the sender, or by the receiver itself based on information from the overall network. As in~\cite{maddah14fundamental, maddah14decentralized, ji14average,ji15order}, we assume that library files and their popularity change at a much slower time-scale compared to the file delivery time-scale. %, and we  neglect the resource requirements associated with the cache-update process. %Hence, the caching phase is sometimes referred to as the {\em placement} or {\em prefetching} phase.
		
	\item[(ii)]{{\em Delivery Phase:}}
	%This phase takes place after completion of the caching phase. During this phase, 
	After the caching phase, only the sender has access to the library and the network is repeatedly  used in a time slotted fashion. At the beginning of each time slot, the sender is informed of the demand realization vector,  denoted by $\dbf = (d_1,\dots,d_K) \in \mathfrak D\equiv [N]^K$, where $d_k \in[N]$  denotes the index of the file requested by receiver $k\in [K]$. 
	
\end{itemize}

\begin{figure}[t]
	\centering
	\includegraphics[width=0.5\linewidth]{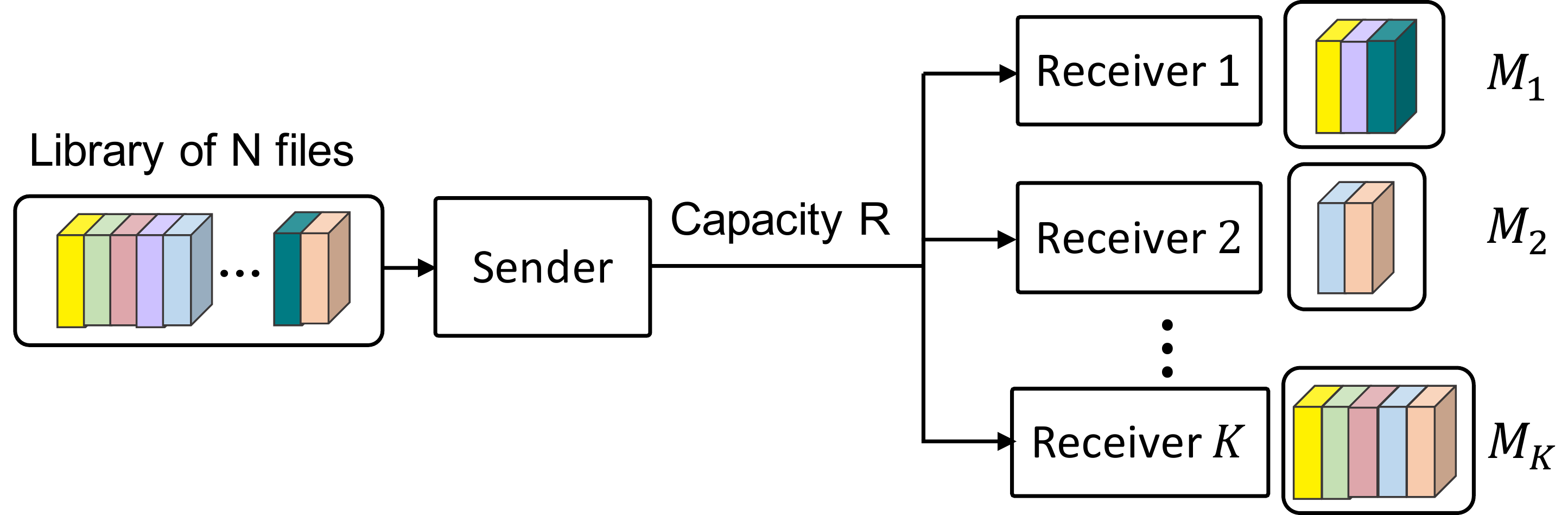}
	\caption{Caching is used for reducing the distortion of requested content in a broadcast network.}
	\label{fig: System}
\end{figure}

The goal of this paper is to design caching and delivery strategies, referred to as {\em caching schemes}, that result in the lowest expected distortion across the network, taken over the source distribution and demand distributions, under the condition that the rate (measured in bits/sample as defined in \eqref{eq: demand rate}) required to satisfy the demand is within a given rate budget $R$, for given receiver cache capacities $M_1,\dots,M_K$.  As a result, when a file from the library is requested, we allow for different versions of the file, encoded at different rates and with different reconstruction distortions, to be delivered to the receivers.

More formally, the caching scheme is composed of the following components:
\begin{itemize}
	\item{\bf Cache Encoder:} 
	The cache encoder at the sender computes the content to be cached at receiver $k\in[K]$, denoted by $Z_k$, using a function $f^{\mathfrak  C}_{k}:  \prod\limits_{n=1}^N{\mathcal W}_n^{F} \rightarrow [1: {  2^{M_kF}}  )$ as $Z_{k} = f^{\mathfrak  C}_{k}\Big( \{W_n^F\}_{n=1}^N\Big) $.
	
	\item{\bf Multicast Encoder:} During the delivery phase, the sender is informed of the demand realization $\mathbf{d} =(d_1, \ldots, d_K)\in\mathfrak D$. The sender uses the function $f^{\mathfrak  M}: {\mathfrak D}   \times \prod\limits_{n=1}^N{\mathcal W}_n^{F}\times \prod\limits_{k=1}^K [1: 2^{M_kF})   \rightarrow \mathcal Y^\star$ to compute and transmit a multicast codeword 
	$Y_{\dbf} = f^{\mathfrak  M}\Big( \dbf ,\, \{W_n^F\}_{n=1}^N,\,   \{Z_{k}\}_{k=1}^K  \Big) $, where we use $\star$ to denote variable length. % over the   link. %,  
	 %where $\mathcal Y^\star$ {\RED denotes the set of finite length sequences,  remove: satisfying the shared-link capacity $R$}. 
	
	\item{\bf Multicast Decoders:} Receiver $k\in[K]$ uses a mapping $g^{\mathfrak  M}_{k} :
	\mathfrak D \times \mathcal Y^\star \times [1: 2^{M_kF}) \rightarrow  \widehat{\mathcal W}_{d_k}^F$ 
	to reconstruct its requested file using its cached content $Z_k$ and the received
	multicast codeword $Y_{\dbf}$, as $\widehat{W}_{d_{k}}^F = g^{\mathfrak  M}_{k}(\dbf, Y_{\dbf} ,Z_{k})$.
\end{itemize}
 For a given demand $\dbf\in\mathfrak D$, the rate transmitted over the shared link, $R_{\dbf}^{(F)}$,  is defined as
	\begin{align}
	R_{\dbf}^{(F)} =  \frac{ \mathbb{E}  [ L(Y_{\dbf})] }{F} ,    \label{eq: demand rate}
	\end{align}
	where $L(Y )$ denotes the length (in bits) of the multicast codeword $Y$,  and the expectation
	is over the source distribution. The expected distortion, over all demands, receivers and the source distribution, is defined as
	$D^{  (F)} =  \mathbb{E} \bigg[ \frac{1}{K}\sum_{k=1}^{K} D_{d_k}( {W}_{d_k}^F,\widehat{W}_{d_k}^F)\bigg]$,  
%\begin{align}
%  D^{  (F)} =  \mathbb{E} \bigg[ \frac{1}{K}\sum_{k=1}^{K} D_{d_k}( {W}_{d_k}^F,\widehat{W}_{d_k}^F)\bigg] ,    \label{eq: exp dist}
%\end{align} 
which is a function of the cached content $\{Z_k\}$ and the multicast codeword $Y_{\dbf}$.
%For example, the expected distortion of Gaussian sources  is given by
 For Gaussian sources we have
\begin{eqnarray} 
D^{(F)}  = %\sum_{ \dbf \in \mathfrak D } \Pi_{\bf d}
 \mathbb{E} \bigg[\frac{1}{K}  \sum_{k=1}^{K}\sigma_{{d}_{k}}^{2} 2^{-2  \,{ \Xi} (M_{k,d_k},\, R_{k, \dbf})}\bigg], \label{averDist}
\end{eqnarray}
where $M_{k,d_k}$ is the size (in bits/sample) of receiver $k$'s cache assigned to storing file $d_k$, $R_{k, \dbf}$  is the total rate (in bits/sample) delivered to receiver $k$ for demand $\dbf$, which we refer to as the {\em per-receiver} rate, and function  $\Xi(.)$ determines the  effective rate available to the receiver useful for  reconstructing its requested file $d_k$, which we refer to as the {\em effective rate function}. 
%$\mathfrak D$ is the set of possible demands, and $\Pi_{\bf d}$ is the probability of demand $\dbf$, i..e,  $\Pi_{\dbf} { \eqdef }\prod_{k=1}^K q_{k,d_{k}}$.

As shown in \cite{maddah14fundamental}, due to the broadcast nature of the wireless transmitter in cache-aided networks, by capitalizing on the spatial reuse of the cached information several different demands can be satisfied with a single coded multicast transmission, resulting in global caching gains. Therefore, for a general caching scheme, the overall rate received by receiver $k\in[K]$ in demand $\dbf$, i.e,. the per-receiver rate $R_{k,\dbf}$, can be different from the total rate multicasted over the shared link by the sender, $R^{(F)}_{\dbf}$. Furthermore, due to the successive refinability of the files, not all messages received and decoded by the receivers are useful for the reconstruction of requested files. Only the cached and received bits that are in sequence determine the reconstruction distortion, translating to the effective rate of $\Xi(M_{k,d_k}, R_{k, \dbf})$ bits/sample.

\begin{definition}\label{def:1}
For a given demand $\dbf$, a distortion-rate-memory tuple $(D,R,M_1,\dots,M_K)$ is {\em achievable} if there exists a sequence of  caching  schemes for cache capacities  $M_1,\dots,M_K$, and increasing file size $F$ such that $\limsup_{F\rightarrow \infty} {  D}^{(F)} \leq D$ and    $\limsup_{F\rightarrow \infty} R_{\dbf}^{(F)} \leq R$. 
\end{definition}

\begin{definition}
The distortion-rate-memory region $ \mathfrak R^*$ is the closure of the set of achievable distortion-rate-memory tuples $(D,R,M_1,\dots,M_K)$, and the optimal distortion-rate-memory 
function  is given by %, $D^*(R,\{M_k\}_{k=1}^K)$, is 
\begin{align}
  D^*(R,\{M_k\}_{k=1}^K) = \inf\Big\{  D: ( D,R,M_1,\dots,M_K) \in  \mathfrak R^*\Big\}.
\end{align}
 \end{definition}
 
 \noindent{Later in Sec.~\ref{sec:RAP-GCCOptimization}, in order to make the optimization problem in \eqref{eq: Het Opt} tractable we use the expected multicast rate	%given in Definition~\ref{def:1} 
 	rather than the multicast rate defined in \eqref{eq: demand rate}, defined as 
${\bar R}^{(F)} = \mathbb{E} [R_{\dbf}^{(F)}] $, 
where the expectation is over the demand distribution. 
}

% ********************************************************************************************************
% UNICAST TRANSMISSIONS
% ********************************************************************************************************

\section{Local Cache-aided Unicast (LC-U) Scheme}\label{sec: Unicast}
In this section, we present LC-U, an achievable scheme that despite its simplicity, serves as a benchmark for caching schemes that exploit coding opportunities during multicast transmissions, which are studied in Sec.~\ref{sec: Multicast}.  Furthermore,  LC-U is useful in refining the multicast content distribution scheme of Sec.~\ref{sec: Multicast}. LC-U determines  the content to be placed in each receiver cache and the multicast codeword that is transmitted over the shared link with rate budget $R$ (bits/sample) for each demand, independently across the receivers. The multicast encoder is equivalent to $K$ independent fixed-to-variable source encoders each depending only on the local cache of the corresponding receiver, resulting in $K$ unicast transmissions. Let ${\bf M}_k=(M_{k,1},\dots,M_{k,N})$ denote the {\em cache allocation} at receiver $k\in[K]$, i.e., the portion of memory designated to storing information from each file. LC-U operates as follows:
 
\begin{enumerate}
\item[(i)] Caching Phase: Receiver $k\in[K]$ computes the optimal cache allocation that minimizes the expected distortion across the network, assuming that it will not receive further transmissions from the sender, i.e., $R_{k, \dbf} = 0$ for any $\dbf\in\mathfrak D$.   Since receivers are not expecting to receive additional refinements during the delivery phase, each receiver caches content independently based on its own demand distribution. Receiver $k \in [K]$ solves the following convex optimization problem
\begin{equation}
\begin{aligned}
&\min
& &  \mathbb E\Big[D_{n} (W_{n}^F , {\widehat W}_{n}^F)\Big]  = \sum_{n=1}^{N} q_{k,n} \sigma_{n}^{2} 2^{-2M_{k,n}} \\
& \text{s.t}
& & \sum_{n=1}^{N} M_{k,n} \leq M_{k}  , \; \;\;M_{k,n} \geq 0, \;\; \forall n \in [N]  
%&&&  M_{k,n} \geq 0, \hspace{1cm} \; \forall n \in [N]  
\end{aligned}\label{eq: LCU-Cache}
\end{equation}
resulting in a cache allocation given as
\begin{equation}
M_{k,n}^{*}=  \left( \log_{2}\sqrt{\frac{2\ln({2}q_{k,n}\sigma_{n}^{2})}{\lambda_{k}^{*}}}\right)^{+}  \label{eq: LCU-CacheSol},
\end{equation}
with $\lambda_{k}^{*}$ such that $\sum_{n=1}^NM_{k,n}^{*}=M_{k}$, and $(x)^+$ is used to denote $\max\{x,0 \}$. The solution admits the well-known reverse water-filling form \cite{Cover}, in which receiver $k$ only stores portions of those files that satisfy $q_{k,n}\sigma_{n}^{2}\leq \frac{\lambda_{k}^{*}}{2\ln{2}}$; hence, $q_{k,n}M_{k,n}^{*} = \min\{\frac{\lambda_{k}^{*}}{2\ln{2}}, \; q_{k,n}\sigma_{n}^{2}\}$, as illustrated in Fig. \ref{fig: WaterFilling}.

 \begin{figure}[ht!]
 	\centering
 	\includegraphics[width=2.1in]{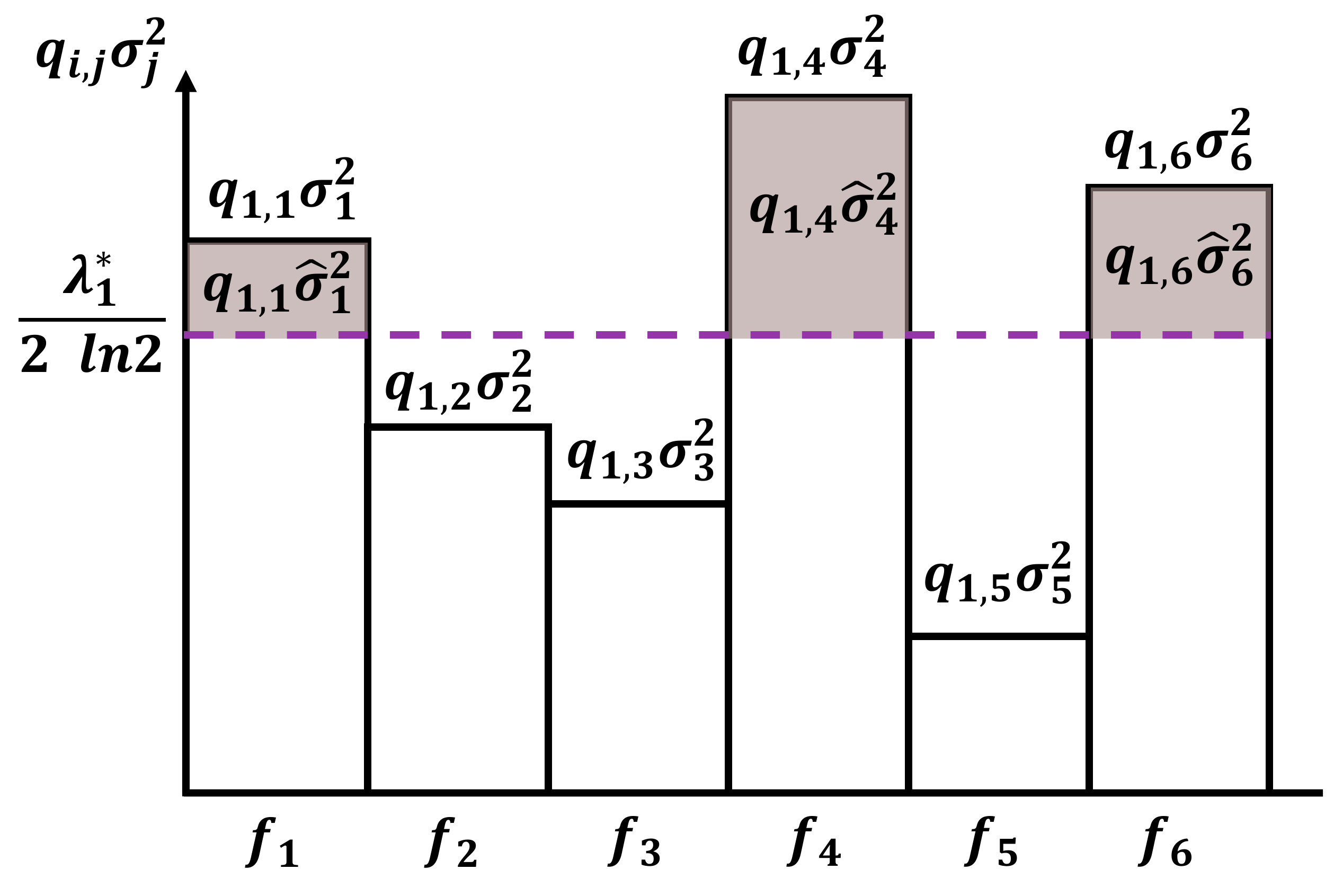}
 	\caption{Cache allocation at receiver $k$ with the LC-U scheme in a network with $N=6$ Gaussian sources.%files, assuming independent Gaussian sources for shared link capacity $R=0$.
 	}
 	\label{fig: WaterFilling}
 \end{figure}

\item[(ii)] Delivery Phase: For a given demand realization $\dbf\in \mathfrak D$, the sender computes the optimal per-receiver delivery rates $\{R_{k,\dbf}\}_{k=1}^K$ jointly across all the receivers in the network by solving the following problem
\begin{equation}
	\begin{aligned}
	&\text{min}
	& & \frac{1}{K} \sum_{k=1}^{K} D_{d_k} (W_{d_k}^F , {\widehat W}_{d_k}^F) = \frac{1}{K} \sum_{k=1}^{K}\sigma_{d_{k}}^{2}2^{-2 (M_{k,d_{k}}^{*}	   +R_{k,\dbf})} \\
	& \text{s.t.}
	& & \sum_{k=1}^{K} R_{k,\dbf} \leq R,  \;\;\; R_{k,\dbf} \geq 0, \;\;\forall k \in [K]
	%&&&  R_{k,\dbf} \geq 0, \hspace{1cm} \forall k \in [K]
	\end{aligned}\label{eq: LCU-Rate}
\end{equation}
which results in  
\begin{equation}
 R_{k,\dbf}^{*}=\left( \log_{2}\sqrt{\frac{2\ln{2}(\sigma_{d_{k}}^{2})}{\gamma_{\dbf}^{*}}}-M_{k,d_{k}}^{*}\right)^{+},
 \end{equation}
with $\gamma_{\dbf}^{*}$ chosen such that $\sum_{k=1}^KR_{k,\dbf}^{*}=R$. 

\begin{remark}    
The caching and delivery strategy described above are such that the receivers cache and receive sequential bits of successively refinable files. Hence, all bits transmitted to the receivers are useful for file reconstruction, i.e., the effective rate delivered to receiver $k\in[K]$ is  \textup{Eff}$(M_{k,d_{k}}^{*},R_{k,\dbf}^*) = M_{k,d_{k}}^{*}+R_{k,\dbf}^*$.
 LC-U is a scalable coding scheme described by two layers, one base layer and one enhancement layer. During the caching phase receiver $k\in[K]$ stores the base layer of file $n\in[N]$ with rate $M_{k,n}$ bits/sample. In the delivery phase, the sender unicasts the enhancement layers of the requested files to the corresponding receivers with rates $\{R_{k,\dbf}\}_{k=1}^K$, using $K$ disjoint multicast encoders.
\end{remark}

\end{enumerate}

In LC-U, the caching process is decentralized and it does not require any coordination from the sender, since receivers fill their caches based on their own preferences, $\{q_{k,n}\}$, and file characteristics, $\{\sigma_{n}^{2}\}$. On the other hand, the requested files are delivered in a centralized manner. Using the cache placements at all receivers, the sender jointly optimizes the per-receiver rates.

% ********************************************************************************************************
% MULTICAST TRANSMISSIONS
% ********************************************************************************************************
 \section{Cooperative Cache-aided Coded Multicast (CC-CM) Scheme}\label{sec: Multicast}
%Generating the multicast codeword while taking into account the broadcast nature of the wireless transmitter and the global information, i.e., the cache content and demands of all receivers, results in more efficient transmissions over the shared link, which, in turn, reduces the expected reconstruction distortion. 
In this section, we present an achievable caching scheme, referred to as CC-CM, that  determines the cache allocations $\{ {\bf M}_k\}_{k=1}^K$, and the per-receiver delivery rates $ \{R_{k,\dbf}\}_{k=1}^K$,  jointly across all receivers and demand realizations $\dbf\in\mathfrak D$, based on the global network information. In addition, CC-CM  uses the broadcast nature of the wireless transmitter. 
%This section presents an achievable caching scheme, which we call the CC-CM  scheme, that  determines the cache allocations $\{ {\bf M}_k\}_{k=1}^K$, and the per-receiver delivery rates $ \{R_{k,\dbf}\}_{k=1}^K$,  jointly across all receivers and demand realizations $\dbf\in\mathfrak D$, and also uses the broadcast nature of the wireless transmitter. Generating the multicast codeword while taking into account the global information, i.e., the cache content and demands of all receivers, results in more efficient transmissions over the shared link, which, in turn, reduces the expected reconstruction distortion.    
%As mentioned in Sec. \ref{sec: Introduction}, it was shown in \cite{maddah14fundamental, maddah14decentralized} that in cache-aided wireless networks jointly designing the cache content and a coded multicast codeword that is simultaneously useful for multiple receivers enables multiplicative caching gains in terms of  the aggregate shared link rate required to deliver the desired per-receiver rates. However, cache-aided delivery schemes available in literature are based on fixed-to-variable source encoding, designed to minimize the aggregate rate on the shared link so that the requested files are reconstructed with zero distortion \cite{maddah14fundamental, maddah14decentralized,  ji15order,yu2017characterizing}, or within pre-specified distortion requirements \cite{timo2016rate,yang2018coded,ibrahim2018coded}.  
 This allows for more efficient use of the shared link compared with LC-U. Similar to LC-U, our goal is to minimize the expected distortion across the network for a given rate budget $R$. To this end, we solve the following problem for cache allocations $\{M_{k,n} \}$ and per-receiver delivery rates $\{R_{k,\dbf}\}$

\begin{equation}
	\begin{aligned}
	&\min
	& &  
	{\mathbb E} \Big[ \frac{1}{K} \sum_{k=1}^{K}\sigma_{d_{k}}^{2}2^{-2\,\Xi(M_{k,d_{k}}	   , R_{k,\dbf})} \Big]\\
	& \text{s.t.}
	& & R_{\ach}\Big( \dbf, \{M_{k,n}\} ,  \{R_{k,\dbf}\} \Big ) \leq R,\quad\forall \dbf\in\mathfrak D\\
	&&& \sum\limits_{n=1}^N M_{k,n}\leq M_k , \;\; M_{k,n},\, R_{k,\dbf} \geq 0, \qquad\; \forall (k,n,\dbf) \in [K]\times[N]\times\mathfrak D %\qquad\forall k\in[K]\\
	%&&&  M_{k,n},\, R_{k,\dbf} \geq 0, \qquad\; \forall (k,n,\dbf) \in [K]\times[N]\times\mathfrak D
	\end{aligned}\label{eq: Main Opt}
\end{equation}
where $R_{\ach}\Big( \dbf, \{M_{k,n}\} ,  \{R_{k,\dbf}\}  \Big ) $ denotes the aggregate multicast rate achieved by the CC-CM scheme for demand realization $\dbf\in\mathfrak D$. The rate $R_{\ach}\Big( \dbf, \{M_{k,n}\} ,  \{R_{k,\dbf}\}  \Big ) $ depends on the architecture of the cache encoder, multicast encoder and multicast decoders described in Sec.~\ref{sec: ProblemSetting} used by CC-CM, and can be difficult to compute in general. In order to make its computation analytically tractable, we focus on a subclass of CC-CM schemes by imposing further restrictions on the per-receiver rates $R_{k,\dbf}$, which could possibly result in a suboptimal solution. Specifically, we assume that for any $(k,\dbf) \in [K]\times\mathfrak D$ the per-receiver rate $R_{k, \dbf}$ is composed of two portions, i.e., $R_{k, \dbf}=  \widetilde R_{k, d_{k}}  + \widehat R_{k, \dbf}$:  (i) a portion, $\widetilde R_{k, d_{k}} $,  delivered  via {\em coded} multicast transmissions, which can be evaluated in a closed-form expression and which depends only on the  receiver-file index pair $(k,d_k)$,  i.e.,  each receiver and its requested file,  and (ii) a portion, $\widehat R_{k, \dbf}$,  delivered via {\em uncoded} muticast transmissions, which depends on the entire demand vector $\dbf$.  The advantage of this approach, as further explained in Sec.~\ref{subsec: RAPP-GCC with SVC}, is that the first multicast portion of the rate, namely $\widetilde{R}_{k, d_{k}}$, can be optimized jointly based on the global information, whereas the second unicast portion, namely $\widehat{R}_{k, \dbf}$, can utilize a solution similar to that of LC-U of Sec. \ref{sec: Unicast} to exhaust  the remaining portion of the  total rate budget of $R$ once  the aggregate  multicast rate is accounted for.  Note that the reconstruction distortion at receiver $k\in[K]$ for demand $\dbf$ is determined by $M_{k,d_k}, \widetilde R_{k,d_k}$ and $\widehat R_{k,\dbf}$, which can vary across different receivers. The optimization problem  in \eqref{eq: Main Opt}, as well as its simplified form, is exponential in the number of receivers since $R_{\ach}\Big(  \dbf,\{M_{k,n}\} ,  \{R_{k,\dbf}\} \Big )$ depends on the demand realization $\dbf\in \mathfrak D \equiv [N]^K$.   Later in Sec.~\ref{subsec:performance of cc-cm}, we simplify \eqref{eq: Main Opt} by replacing the exponential number of per-demand rate constraints with an average rate contraint.

 Throughout this paper, we use \emph{aggregate coded rate} to refer to the overall rate sent over the shared link through coded multicast transmissions, which is a function of  the per-receiver coded rates $\{\widetilde R_{k, d_{k}}\}$. Additionally, we use \emph{aggregate uncoded rate} to refer to the overall rate transmitted through uncoded transmissions, which is a function of  the per-receiver uncoded rates $\{ \widehat R_{k, \dbf}\}$.  While the aggregate uncoded rate can be upper bounded by the sum rate $ \sum_{k=1}^K \widehat R_{k, \dbf}$, the aggregate coded rate depends on the specific scheme adopted for the coded multicast transmission and its multiplicative coding gains.

In the remainder of this paper, in order to characterize the aggregate coded rate, and to implement the caching phase and the portion of the delivery phase corresponding to coded multicast transmissions, we adopt a generalization of the cache-aided coded multicast delivery scheme proposed in  \cite{ji15order}, which we refer to as RF-GCC. %{\em Random Fractional (RF) caching  with Greedy Constrained Coloring (GCC)}. 
Then, we use this scheme and a variation of the LC-U  of Sec.~\ref{sec: Unicast}  to quantify $R_{\ach}\Big( \dbf, \{M_{k,n}\} ,  \{R_{k,\dbf}\}  \Big )$ in problem \eqref{eq: Main Opt}.

In Sec.~\ref{sec:RAP}, we first describe RF-GCC, obtained by generalizing {\em Random Aggregate Popularity caching with Greedy Constrained Coloring (RAP-GCC)} \cite{ji15order} to a setting with heterogeneous cache sizes, where receivers are interested in receiving different-length portions of the same file that map to the possibly different per-receiver reconstruction distortions resulting from problem \eqref{eq: Main Opt}.  We then derive an upper bound on the aggregate coded rate achieved with this scheme. 
We start Sec.~\ref {sec:RAP-GCCOptimization} by describing how RF-GCC can be adopted by   CC-CM to fill receiver caches and to deliver the coded portion of the multicast   codeword. We then use the upper bound on the rate achieved by RF-GCC   to characterize $R_{\ach}\Big( \dbf, \{M_{k,n}\} ,  \{R_{k,\dbf}\}  \Big ) $, and to solve optimization \eqref{eq: Main Opt}. %the rate-distortion-memory trade-off of the proposed CC-CM scheme.  

%In the remainder of this paper, in order to quantify the aggregate coded rate, and to implement the caching phase and the portion of the delivery phase corresponding to coded multicast transmissions, we adopt a generalization of the {\em Random Aggregate Popularity (RAP) caching with Greedy Constrained Coloring (GCC)}, or {\em RAP-GCC}, scheme proposed in  \cite{ji15order}, which we refer to as {\em Random Fractional (RF) caching with GCC}, or {\em RF-GCC}. Then, we use this scheme and a variation of the LC-U  of Sec.~\ref{sec: Unicast}  to quantify $R_{\ach}\Big( \dbf, \{M_{k,n}\} ,  \{R_{k,\dbf}\}  \Big )$ in problem \eqref{eq: Main Opt}.

% In Sec.~\ref{subsec: comparison to QD}, we discuss how   RF-GCC    generalizes prior work in  \cite{yang2018coded} to a setting with file heterogeneity, and in Sec.~\ref {subsec: RAPP-GCC with SVC} we describe how RF-GCC can be adopted by   CC-CM   to deliver the coded portion of the multicast   codeword. In Sec.~\ref{sec:RAP-GCCOptimization}, we use the upper bound on the rate achieved by   RF-GCC   to characterize $R_{\ach}\Big( \dbf, \{M_{k,n}\} ,  \{R_{k,\dbf}\}  \Big ) $, and to quantify the performance of the proposed CC-CM scheme.  

%\section{\BLUE Random Fractional with Greedy Constrained Coloring (RF-GCC) Scheme}

\section{Coded Multicast Delivery Through RF-GCC Scheme}\label{sec:RAP}
%This section presents a generalization of the RAP-GCC scheme proposed in \cite{ji15order} to a heterogeneous system, where receivers have different cache sizes, and they receive different-length versions of the files in the library.  

This section presents RF-GCC, a generalization of RAP-GCC  proposed in  \cite{ji15order}. In Sec.~\ref{sec:RAP-GCCOptimization}, we discuss how RF-GCC is adopted by CC-CM to implement the caching phase and the portion of the delivery phase corresponding to coded multicast transmissions. Recall that per Sec.~\ref{sec: Multicast}, the coded multicast transmission corresponds to rate $\widetilde{R}_{k,d_k}$ delivered to receiver $k\in[K]$.
	
In the system model considered in \cite{ji15order}, the files are generated from sources with the same distribution, they are requested by all receivers according to the same demand distribution, and the goal is to design a caching scheme that minimizes the expected multicast rate. Therefore, in the setting of \cite{ji15order}, the authors propose a scheme where the caching phase is designed only based on the {\em aggregate demand distribution}, i.e., the probability that a file is requested by at least one user. %, and they refer to their caching strategy as Random Aggregate Popularity (RAP). 
In this paper, our goal is to minimize the expected distortion and we adopt a generalization of this caching policy where the caches are filled based on not only the demand distributions but also the distortion-rate functions of the files. To this end, we use the more general term {\em Random Fractional (RF)} caching rather than RAP caching used in \cite{ji15order}.

Consider a cache-aided system with a library of $N$ files with length $\tau$ bits\footnote{We use $\tau$ instead of the conventional notation $F$ in  \cite{maddah14decentralized,maddah14fundamental,ji15order} to avoid confusion with the number of source-samples in Sec.~\ref{subsec: RAPP-GCC with SVC}.}  indexed by $\{ 1,\dots,N\}$ and $K$ receivers  $\{1,\dots,K\}$, where receiver $k\in[K]$ has a cache of size $\mu_k\tau$ bits.  For each file $n\in[N]$ in the library, there are $K$ different fixed-size {versions} available one for each receiver, such that {\em version} $k$ of file $n$ is composed of the {\em first} $\length_{k,n}\factor$ bits of file $n$.  For two indices $k_1$ and $k_2$, we say that version $k_1$ of file $n$, with length $\length_{k_1,n}\factor$, is a {\em degraded version} of version $k_2$ of file $n$, with length  $\length_{k_2,n}\factor$, if $\length_{k_1,n}\leq\length_{k_2,n}$. Receivers request files from the library following the demand distributions described in Sec.~\ref{sec: ProblemSetting}, and when receiver $k\in[K]$ requests file $n\in[N]$, the sender delivers version $k$ of file $n$ with length $\Omega_{k,n}$. Note that a version of a file is composed of a fixed number of successive bits, which corresponds to a file having a predetermined reconstruction distortion (e.g., video playback quality)\footnote{In our setting, the different versions of a file can correspond to different numbers of enhancement layers in successively refinable  compression (further explained in Sec.~\ref{subsec: RAPP-GCC with SVC}),  or could correspond to different-length portions of the same document.}. Throughout the remainder of this section we assume that the version lengths $\{\Omega_{k,n}\}$ are fixed, and later in Sec.~\ref{sec:RAP-GCCOptimization} we determine the optimal version lengths based on \eqref{eq: Het Opt}.

\begin{remark}\label{rmk:compare to QD} 
The setting considered in this section is similar to the one considered in \cite{yang2018coded}, where receivers have predefined distortion requirements. In \cite{yang2018coded}, the objective is to design an efficient caching scheme that minimizes the worst-case delivery rate over the shared-link for a given set of receiver cache capacities and distortion requirements. In our setting, the version lengths $\{\Omega_{k,n}\tau\}$ can be interpreted as receiver distortion requirements, which further generalizes the problem in \cite{yang2018coded} to each receiver having different distortion requirements for each file in the library. Differently from \cite{yang2018coded}, as defined in \eqref{eq: Main Opt}, our ultimate goal  is to minimize the expected distortion across the network for a given set of cache capacities and a given shared-link rate budget. As a means to solving the general problem in \eqref{eq: Main Opt}, our proposed solution in this subsection extends that of \cite{yang2018coded} to a setting with heterogeneity across files in addition to across receivers. 
	
\end{remark} 

We solve the problem defined in \cite[Sec. II]{ji15order}, by finding an upper bound on the rate-memory trade-off in cache-aided networks, for the setting described above, where degraded versions of files with different lengths are delivered to the receivers. In the following, we (i) describe a decentralized caching scheme in Sec.~\ref{subsec:description}, and (ii) characterize its achievable rate for a given demand in Theorem~\ref{thm:demand}, and on average over all demands in Theorem~\ref{thm:general}.

\subsection{Scheme Description}\label{subsec:description}
 As in conventional caching schemes, a fractional cache encoder divides each file into packets and determines the subset of packets from each file that are  stored in each receiver cache. For each demand realization in the delivery phase, the multicast encoder generates a multicast codeword  by computing an {\em index code} based on a coloring of the index coding conflict graph \cite{birk1998informed,IndexCoding}. The  RF-GCC scheme operates as follows:
 % As in conventional caching schemes, a fractional cache encoder divides each file into packets and determines the subset of packets from each file that are  stored in each receiver cache. To ensure robustness to system dynamics, the cache encoder fills the caches in a random fashion, usually referred to as {\em decentralized} caching in the literature, where the packets to be cached are selected according to a given caching distribution. The caching distribution is designed with respect to the objective of the application in mind, e.g., to minimize the expected delivery rate based on the demand distributions as in \cite{ji15order}. For each demand realization in the delivery phase, the multicast encoder generates a multicast codeword constructed as a linear combination of the requested packets that are not locally available, such that each receiver can losslessly recover the corresponding version of its requested files. The multicast encoder generates this codeword   by computing an {\em index code} based on a coloring of the index coding conflict graph \cite{birk1998informed,IndexCoding}. The  RF-GCC scheme operates as follows:
\begin{enumerate}
	\item[(i)] Caching Phase: All the versions of the library files are partitioned into equal-size packets of lengths $ T$ bits. The cache encoder is characterized by $K$ vectors,  $\pbf_k=(p_{k,1},\ldots, p_{k,N})$, $k = 1,\dots,K$, referred to as the {\em caching distributions}, such that $  p_{k,n}\in[0,1/\mu_k]$ and $\sum_{n=1}^N p_{k,n} = 1$, for any $k\in[K]$. Element $p_{k,n}$ represents the portion of receiver $k$'s cache capacity that is assigned to storing packets from version $k$ of file $n\in[N]$. Receiver $k\in[K]$  selects and stores a subset of $ p_{k,n}\cache_k \factor /T$ distinct packets from version $k$ of file $n$, uniformly at random. The caching distributions, $\{\pbf_1,\dots,\pbf_K\}$, are optimally designed based on an objective function, for example to minimize the rate of the corresponding index coding delivery scheme as in \cite{ji15order}, or to minimize the expected network distortion as in Sec.~\ref{sec:RAP-GCCOptimization} of this paper. In the following, we denote by $\Cbf =\{\Cbf_1,\dots,\Cbf_K \}$ the packet-level cache configuration, where $\Cbf_k$  denotes the set of packets cached at receiver $k\in[K]$, which correspond to the packets from version  $k$ of all library files.

\item[(ii)] Delivery Phase: For a given demand realization $\dbf$, we denote the packet-level
demand realization by $\Qbf = \{\Qbf_1,\dots,\Qbf_K \}$, where $\Qbf_k$ denotes the set of packets from the file version requested by receiver  $k\in[K]$, i.e., version  $k$ of file $d_k$ with length $\Omega_{k,d_k}\tau$, that are not cached at it. In order to determine the set of packets that need to be delivered, the sender constructs an index coding {\em conflict graph}, which is the complement of the side information graph as described in  \cite{birk1998informed,IndexCoding}. For a given packet-level cache configuration $\Cbf$ and demand realization $\Qbf$, the conflict graph, denoted by $\mathcal H_{\Cbf,\Qbf}$, is constructed as follows:  
\begin{enumerate}
	\item For each requested packet in $\Qbf$, there is a vertex $v$ in the graph uniquely identified by the label $\{\alpha(v),\beta(v),\eta(v)\}$, where $\alpha(v)$ indicates the packet identity associated to $v$, $\beta(v)$ is the receiver requesting it and $\eta(v)$ is the set of all receivers that have cached  the packet.
	\item For any two vertices $v_1$, $v_2$, we say that vertex $v_1$ interferes with vertex $v_2$ if: $1)$ the packet associated with $v_1$, $\alpha(v_1)$, is not in the cache of the receiver associated with  $v_2$, $\beta(v_2)$; and if $2)$ $\alpha(v_1)$ and $\alpha(v_2)$ do not represent the same packet. There exists an undirected edge between $v_1$ and $v_2$ if $v_1$ interferes with $v_2$ or if $v_2$ interferes with $v_1$.
\end{enumerate}
Given a valid vertex coloring\footnote{A valid vertex coloring is an  assignment of colors to   vertices such that no two adjacent vertices are assigned the same color.} of the conflict graph $\mathcal H_{\Cbf, \Qbf}$, the multicast encoder generates the multicast codeword by concatenating the XOR of the packets with the same color. A chromatic number {\em index code} for this graph results from generating the multicast codeword based on the valid coloring that results in the shortest codeword.  Computing the index code based on graph coloring is NP-complete and quantifying its performance can be quite involved. In order to quantify the achievable rate, as in  \cite{ji15order},  we adopt a greedy approximation of the algorithm referred to as Greedy Constrained Coloring (GCC), which has polynomial-time complexity in the number of receivers and packets.  Due to space limitations, we refer the reader to  \cite[Algorithms 1 and 2]{ji15order} for the pseduo code of GCC. This coloring results in a possibly larger multicast codeword compared to the chromatic number index code, but as shown in \cite{ji15order}, for very large block lengths ($\tau\rightarrow\infty$),  its achievable rate:  (i) can be evaluated in a closed-form expression, and (ii) it provides a tight upper bound on the rate achieved with the chromatic number index code (i.e., it is asymptotically order-optimal). %However,  the performance of RF-GCC degrades for finite block lengths $\tau$. Hence, when operating in the finite block length regime, one can design improved greedy coloring algorithms as the one proposed in \cite{ji15efficient}, that have polynomial-time complexity and that exploit the structure of the conflict graph to recover a significant amount of the multiplicative caching gains observed in the asymptotic regime. 
 
  \end{enumerate}

Let $R^{C}\Big(\dbf, \{\mu_k \}, \{\pbf_{k}\}, \{{\length}_{k,n}\}  \Big)$  denote the asymptotic coded multicast rate achieved by  RF-GCC, as $\factor\rightarrow\infty$, for a given demand $\dbf$, caching distributions $\{\pbf_k\}$ 
and file version lengths $\{\length_{k,n}\factor\}$.  As in caching literature, $R^{C}\Big(\dbf, \{\mu_k \},  \{\pbf_{k}\}, \{{\length}_{k,n}\}  \Big)$ is defined as the limiting value ($\factor\rightarrow\infty$) of the length (in bits) of the multicast codeword nominalized by $\factor$. Next, we provide an upper bound on the achievable rate $R^{C}\Big(\dbf, \{\mu_k \},  \{\pbf_{k}\}, \{{\length}_{k,n}\}  \Big)$, which is used in Sec.~\ref{sec:RAP-GCCOptimization} to solve optimization problem \eqref{eq: Main Opt}, and to derive the optimal values  $\{M_{k,n}^*\}$,  $\{\widetilde R_{k,n}^*\}$ and  $\{\widehat R_{k, \dbf}^*\}$.

 \subsection{Achievable Rate}
\label{sec:Achievable Rate}
The following theorems provide closed-form upper bounds on the delivery rate. % $R^C\Big(\dbf, \{\mu_k \},  \{\pbf_{k}\}, $ $\{\length_{k,n}\}  \Big)$.   
Specifically, Theorem \ref{thm:demand} characterizes the achievable rate  for demand $\dbf\in\mathfrak D$, i.e., $R^C\Big(\dbf, \{\mu_k \},  \{\pbf_{k}\}, $ $\{\length_{k,n}\}  \Big)$, while Theorem \ref{thm:general} upper bounds the expected rate over all demand realizations, denoted by ${\bar R}^C\Big(\{\qbf_k\} , \{\mu_k \},$ $\{\pbf_{k}\},\{\length_{k,n}\}  \Big)$.

 \begin{theorem}
	\label{thm:demand}
	In a network with $K$ receivers and $N$ files, for a given demand    $\dbf\in\mathfrak D$,  a given set of cache capacities $\{\mu_k\}_{k=1}^K$ and caching distributions $\{\pbf_{k} \}_{k=1}^K$,  
	the asymptotic coded multicast rate required to deliver the requested file versions with length $\{\length_{k,d_k} \factor\}_{k=1}^K$,  is upper bounded as 
	\begin{align}  
	%  \limsup_{\factor\rightarrow\infty} 
	R^C \Big(\dbf, \{\mu_k \}  ,\{\pbf_{k}\},\{{\length}_{k,n}\} \Big)  \leq  \min \bigg\{      \Psi_{\dbf}^{(1)}\Big(  \{\mu_k \} ,\{\pbf_{k}\},\{{\length}_{k,n}\}    \Big),   \Psi_{\dbf}^{(2)}\Big(   \{\mu_k \}  ,\{\pbf_{k}\},\{{\length}_{k,n}\}  \Big) \bigg\},\notag
	\end{align}
	\noindent where
	\begin{align}
	&  \Psi_{\dbf}^{(1)}\Big( \{\mu_k \} , \{\pbf_{k}\},\{{\length}_{k,n}\}   \Big) =   \sum_{i=1}^{K} \sum_{\ell=1}^{ K-i+1}   \sum_{\mathcal K_{\ell} \subseteq    \{\order_i,\dots,\order_{K}\}  }  \;  
		\Big( \length_{\order_i,d_{\order_i}}  -  	 \length_{\order_{i-1},d_{\order_{i-1}} }    \Big)   \,  \max\limits_{k\in\mathcal K_\ell}	{ \lambda_i }(\mathcal K_\ell,k,d_{k}) ,  \notag\\
	&   \Psi_{\dbf}^{(2)} \Big(  \{\mu_k \}  ,\{\pbf_{k}\},\{{\length}_{k,n}\}\Big)=   \sum_{n =1}^N  \mathbbm{1} \{n \ni \dbf\} \Big( \max\limits_{k: d_k = n}  \, {\length}_{k,n}  - \min\limits_{k: d_k = n} p_{k,n}\mu_k\Big) ,   \label{eq:psi 2} \\
	& {  \lambda_i} (\mathcal K_\ell,k,n)   = (1-p^c_{k,n})   \prod\limits_{u\in \mathcal{K}_{\ell}\backslash \{k\}}p^c_{u,n} \prod\limits_{u\in {  \{\order_i,\dots,\order_{K}\}}\setminus  \mathcal{K}_{\ell}}{(1-p^c_{u,n})},\label{eq:lambda thm}\\
	& p^c_{k,n} =p_{k,n} \frac{ \mu_k}{ \length_{k,n}} , \label{eq: pc is prob}
	\end{align}
	%and for a given demand $\dbf$ and receiver subset $\mathcal K_\ell \subseteq [K]$, receiver $k^*$ is the one requesting the shortest version of file, i.e. $k^*  = \argmin\limits_{k\in\mathcal K_\ell} \Omega_{k, d_k}$.  
	where $\mathcal{K}_{\ell}$ denotes a given set of $\ell$ receivers, and for a given demand $\dbf$, $\order_1,\dots,\order_K$ denotes an ordered permutation of receiver indices such that $\Omega_{\order_1,d_{\order_1}}\leq\dots\leq\Omega_{\order_K, d_{\order_K}}$, where ${\length}_{\order_0,\order_{d_0}}=0$ and $\{\order_1,\order_0  \} = \emptyset$. In \eqref{eq: pc is prob}, $p^c_{k,n}$ denotes the probability that a packet from version $k\in[N]$ of file $n\in[N]$ is cached at receiver $k$. We use $i\ni \xbf$ to indicate that $i$ is one of the elements of  $\xbf$.
\end{theorem}

\begin{proof}
	The proof is given in Appendix \ref{App:demand}.
\end{proof}

By averaging over all possible demand realizations $\dbf\in\mathfrak D$ we obtain the following result.
% ***************************************************************************
% Theorem 2
\begin{theorem}
	\label{thm:general}
	In a network with $K$ receivers and $N$ files, for a given set of demand distributions   $\{\qbf_k\}_{k=1}^{K}$,  cache capacities $\{\mu_k \}_{k=1}^K$, and caching distributions $\{\pbf_{k} \}_{k=1}^K$,  
	the asymptotic expected coded multicast rate required to deliver the requested file versions with length $\{\length_{k,n} \factor \}_{(k,n)\in[K]\times[N]}$,  is upper bounded as 
	
 \scalebox{0.9}{\parbox{\linewidth}{%
	\begin{equation} %\label{eq:demanrate}
	{\bar R}^C \Big(\{\qbf_{k}\},\{\mu_k \} ,\{\pbf_{k}\},\{{\length}_{k,n}\}   \Big)  \leq  \min \bigg\{    
	\bar\Psi^{(1)}\Big(  \{\qbf_{k}\},\{\mu_k \} , \{\pbf_{k}\}, \{  \length_{k,n} \}   \Big),   \bar\Psi^{(2)}\Big(  \{\qbf_{k}\},\{\mu_k \} , \{\pbf_{k}\}, \{  \length_{k,n} \}  \Big) \bigg\},\notag
	\end{equation}
}}
	\begin{align}
	& \text{where\quad} {\bar \Psi}^{(1)}\Big( \{\qbf_{k}\}, \{\mu_k \} , \{\pbf_{k}\}, \{  \length_{k,n} \}     \Big) \notag\\
	&\qquad\qquad= 
	 \sum_{i=1}^{K} \sum_{\ell=1}^{K-i+1}  \sum_{n =1}^{N}\sum_{\mathcal K_{\ell} \subseteq \{\order_{i}^*,\dots,\order_K^*  \}}    \sum_{k \in {\mathcal K}_\ell  }
	 \Big(\length_{\order_i^*}^* -  \length_{\order_{i-1}^*}^*  \Big) \lambda_i(\mathcal K_\ell,k,  n)  \,\Gamma_i(\mathcal K_\ell,k,n),  \label{eq:psi 1 exp} \\
	&  {\bar \Psi}^{(2)} \Big(  \{\qbf_{k}\},\{\mu_k \} , \{\pbf_{k}\}, \{  \length_{k,n} \}  \Big) = \sum_{n=1}^N \Big(1-\prod_{k=1}^K (1-q_{k,n})\Big)\Big(\max\limits_{k\in[K]}   \length_{k,n}  -  \min\limits_{k\in[K]}  p_{k,n}\mu_k \Big) 
	,   \label{eq:psi 2 exp} \\
	&   \Gamma_i({\mathcal K_\ell,k,n} ) = \mathbb P\Big( (k,n)= \argmax\limits_{ ( s,t): s\in\mathcal K_\ell, \,t = f_s  }\,\lambda_i(\mathcal K_\ell,s,t)    \Big),\label{eq:gammaa}\\
	& \length^*_{k} = \max_{n\in[N]} \length_{k,n} ,
	\end{align}
and with $\lambda_i  (\mathcal K_\ell,k,n) $  defined in \eqref{eq:lambda thm}, and where $\order_1^*,\dots,\order_K^*$ denotes an ordered permutation of  receiver indices $\{1,\dots,K  \}$ such that $\length_{\order_1^*}^*\leq\dots\leq\Omega_{\order_K^*}^*$. In \eqref{eq:gammaa},  $\fbf$ denotes the $\ell$-dimensional sub-vector of demand $\dbf$ corresponding to receivers in set $\mathcal K_\ell$, and $\Gamma_i(\mathcal K_\ell,k,n) $ denotes the probability that file $n \ni \fbf$ requested by receiver $k\in\mathcal K_\ell$ maximizes  the quantity $\lambda_i(\mathcal K_\ell,s,t) $. 
  
\end{theorem}

\begin{proof}{}
	%The proof is given in \cite[Appendix B]{onlineVer} following steps similar to those in \cite[Appendix A]{ji15order}.
The proof is given in Appendix~\ref{App:general}. %or The proof is given in \cite[Appendix B]{onlineVer}.
\end{proof}
 
\subsection{Special Cases for  RF-GCC}\label{subsec: comparison to QD}
 In this section, we focus on two specialized settings with symmetry across the library files or across receivers. In Sec.~\ref{subsec:Symmetrical file}, we describe how, under file-symmetry, the proposed RF-GCC scheme is applicable to the problem studied in \cite{yang2018coded}, and Sec.~\ref{subsec:Symmetrical rec} considers symmetry across receivers, which is used in Sec.~\ref{subsec:performance of cc-cm} to solve the optimization problem in \eqref{eq: Main Opt}.

\subsubsection{\bf Symmetry Across Files}\label{subsec:Symmetrical file}
As explained in Remark~\ref{rmk:compare to QD}, the  RF-GCC proposed in Sec.~\ref{subsec:description} can be adopted for the problem studied in \cite{yang2018coded}. The network  in  \cite{yang2018coded} is composed of $N$ independent files and $K$ receivers with cache sizes $\{\mu_1,\dots,\mu_K\}$. Each receiver has a preset distortion requirement, $\{D_1,\dots,D_K\}$, i.e., any of the library files requested by receiver $k$ need to be delivered with distortion less than $D_k$, and the objective is to characterize the rate-memory trade-off for the worst-case demand. Then, for a given distortion-rate function, %for example the one given for Gaussian sources in Sec.~\ref{subsec:source}, 
the distortion requirements of receivers can be mapped to a given set of minimum compression rates. The minimum compression rates are equivalent to the normalized (by  constant $\tau$) version lengths $\{\Omega_{k,n}\}$ defined in Sec.~\ref{sec:RAP}, when $\Omega_{k,1}=\dots=\Omega_{k,N}$ for any $k\in[K]$. Therefore, the setting considered in \cite{yang2018coded} is a specialization of our network model in Sec.~\ref{subsec:description} to the case where each receiver is interested in getting equal length versions of the files in the library.

 Note that we characterize the rate-memory trade-off by deriving an upper bound on the rate of  RF-GCC for any given demand in Theorem~\ref{thm:demand}, from which we then characterize the average rate-memory trade-off in Theorem  \ref{thm:general}, while this trade-off is only provided for the worst-case scenario in  \cite[Sec IV]{yang2018coded}.    Specializing Theorem~\ref{thm:demand} to equal version lengths  leads to Corollary \ref{thm:demand QD} below. The worst-case rate-memory trade-off is given by  $\max\limits_{\dbf\in\mathfrak D}R^C \Big(\dbf, \{\mu_k \}  ,\{p_{k}\},\{{\length}_{k}\} \Big) $, and the average  trade-off can be derived as in  %\cite[Appendix B]{onlineVer} 
 Appendix~\ref{App:general} by taking expectation of the rate over all demands.

% ***************************************************************************** 
% Corllary 1
\begin{corollary}\label{thm:demand QD}
	In a network with $K$ receivers and $N$ files, for a given demand realization  $\dbf\in\mathfrak D$, and a given set of cache capacities $\{\mu_k\}$ and caching distributions with parameters $\{p_{k} \}_{k=1}^K$,  
	the asymptotic coded multicast rate required to deliver the requested file versions with length $\{\length_{k} \factor\}_{k=1}^K$,  is upper bounded as 
	\begin{align}  
	%  \limsup_{\factor\rightarrow\infty} 
	R^C \Big(\dbf, \{\mu_k \}  ,\{p_{k}\},\{{\length}_{k}\} \Big)  \leq  \min \bigg\{      \Psi_{\dbf}^{(1)}\Big(  \{\mu_k \}  ,\{p_{k}\},\{{\length}_{k} \}   \Big),   \Psi_{\dbf}^{(2)}\Big(   \{\mu_k \}  ,\{p_{k}\},\{{\length}_{k}\}\Big) \bigg\},
	\end{align}
	\noindent where
	\begin{align}
	&  \Psi_{\dbf}^{(1)}\Big( \{\mu_k \} , \{p_{k}\},\{{\length}_{k}\}   \Big) = \sum_{i=1}^{K} \Big( \length_{\order_i}  -  	 \length_{\order_{i-1}}     \Big) \sum_{\ell=1}^{K-i+1} \sum_{\mathcal K_{\ell} \subseteq   \{\order_i,\dots,\order_{K}\}}  
	  \max\limits_{k\in\mathcal K_\ell}	\lambda_i (\mathcal K_\ell,k) ,  \label{eq:psi 1 QD} \\
	&   \Psi_{\dbf}^{(2)} \Big(  \{\mu_k \}  ,\{p_{k}\},\{{\length}_{k}\} \Big)=   
	\sum_{n =1}^N  \mathbbm{1} \{n \ni \dbf\} \Big( \max\limits_{k: d_k = n}  \, {\length}_{k} - \min\limits_{k: d_k = n}  \, {\mu}_{k}p_k \Big)  ,   \label{eq:psi 2 QD} \\
	& \lambda_i (\mathcal K_\ell,k)   = (1-p^c_{k})\prod\limits_{u\in \mathcal{K}_{\ell}\backslash \{k\}} p^c_{u} \prod\limits_{u\in\{\order_i,\dots,\order_{K}\}\setminus  \mathcal{K}_{\ell}}{(1-p^c_{u})},\label{eq:lambda thm QD}\\
	& p^c_{k} =p_{k} \frac{ \mu_k}{ \length_{k}}\label{eq: pc is prob QD}	
	\end{align}	
	where $\mathcal{K}_{\ell}$ denotes a given set of $\ell$ receivers, and $\order_1,\dots,\order_K$ denotes an ordered permutation of receiver indices such that $\Omega_{\order_1}\leq\dots\leq\Omega_{\order_K}$. In \eqref{eq: pc is prob QD}, $p^c_{k}$ denotes the probability that a packet from version $k\in[N]$ of any file $n\in[N]$ is cached at receiver $k$.
	\end{corollary}

\begin{remark}
For the setting considered in \cite{yang2018coded} with $\length_1\leq\dots,\leq\length_K$, we have observed  that when $N\geq K$, the worst-case delivery rate computed based on Corollary~\ref{thm:demand QD} is equal to the rate provided in \cite[Theorem 5]{yang2018coded}. Our numerical results show slight improvement in delivery rate compared to the rate  in \cite[Theorem 5]{yang2018coded} for the less common setting of $N<K$.

\end{remark}

\subsubsection{\bf Symmetry Across Receivers}~\label{subsec:Symmetrical rec}
Consider a network with symmetric receivers where all receivers have equal-size caches and request files according to the same demand distribution, i.e. $M_{k} = M$, $q_{k,n} = q_{n}$,  for all $(k,n) \in [k]\times[N]$.  In this network, it is immediate to verify that the optimal caching
distributions $\{\pbf_k \}$, and the corresponding cache allocations
$\{ M_{k,n}\}$ are uniform across all the receivers, i.e., $p_{k,n} = p_n$ and $M_{k,n} = M_n$,  for all $(k,n) \in [k]\times[N]$.  Furthermore, all receivers have the same storing range for file $n\in[N]$, i.e., $\length_{k,n} = \length_{n}$. %  or equivalently, they have  the same per-receiver coded rate $\widetilde R_{k,n} =\widetilde R_{n}$. 

 The following theorem characterizes the  asymptotic expected coded multicast rate achieved with the  RF-GCC scheme in this symmetric setting, and provides a tighter upper bound on the expected multicast rate compared to the one resulting from  specializing Theorem~\ref{thm:general} to a setting with symmetric receivers. Theorem~\ref{thm:Sym user} generalizes the results in \cite{ji15order}, and characterizes the expected coded multicast rate achieved in a network composed of symmetric receivers and non-symmetric files with unequal popularities $\qbf$ and lengths $\{\length_n \tau \}$.

%In this symmetric setting, constraint \eqref{eq:constraint relaxed} becomes ${\bar R}^{C}\Big( \qbf ,M ,\pbf, \{\length_n\} \Big )+   {\mathbb E}\Big [ \sum\limits_{k=1}^{K} \widehat R_{k, \dbf} \Big] \leq R$, where ${\bar R}^{C}\Big( \qbf,M, \pbf, \{\length_n\} \Big )$,  the asymptotic expected  coded multicast rate achieved with the  RF-GCC scheme is given next.

\begin{theorem}\label{thm:Sym user}
	In a network with $K$ symmetric receivers and $N$ files, demand distribution   $\qbf$, cache capacity $\mu$ and caching distribution $\pbf$, the asymptotic expected coded multicast rate required to deliver the requested file versions with length $\{\length_{n} \factor \}_{n=1}^{ N}$,  is upper bounded as 
	\begin{equation} \label{eq:demanrate special case}
	{\bar R}^C \Big(\qbf,\mu,\pbf,\{{\length}_{n}\}   \Big)  \leq  \min \bigg\{    
	\bar\Psi^{(1)}\Big(  \qbf,\mu,\pbf,\{{\length}_{n}\}     \Big),   \bar\Psi^{(2)}\Big(\qbf,\mu,\pbf,\{{\length}_{n}\}    \Big)  \bigg\},
	\end{equation}
	\noindent where
	\begin{align}
	& {\bar \Psi}^{(1)}\Big(\qbf,\mu,\pbf,\{{\length}_{n}\}    \Big) = 
	\sum_{i=1}^{N} (\length_{\zeta_i} -  \length_{\zeta_{i-1}}  ) \sum_{\ell=1}^{{\widetilde K}_i}    \binom{{\widetilde K}_i }{\ell}\sum_{  n\in\{\zeta_i,\dots,\zeta_N \}}
	\Gamma_i({\widetilde K}_i,\ell,n) \, \lambda({\widetilde K}_i,\ell,n)  	,   \label{eq:psi 1 special case} \\
	& {\bar \Psi}^{(2)}\Big(\qbf,\mu,\pbf,\{{\length}_{n}\}    \Big) =   \sum_{n=1}^N \Big(1-  (1-q_{n})^K\Big)\Big(    {\length}_{n}  -  p_{n}\mu\Big)	,   \label{eq:psi 2 special case} \\
	&  \lambda(K,\ell,n)=  (   p_n^c  )^{\ell-1} ( 1-p_n^c   )^{K-\ell+1} , \label{eq:lambda special case}  \\
	&   \Gamma_i({\ell,n} ) = \mathbb P\Big( n= \argmax\limits_{t\in\mathcal F_\ell} \;\; (   p_t^c  )^{\ell-1} ( 1-p_t^c   )^{K-\ell+1}    \Big),\label{eq:gammaa special case}\\
	& p^c_{n} =p_{n} \frac{ \mu}{ \length_{n}} , \label{eq: pc is prob special case} 
	\end{align}
	where   $\zeta_1,\dots,\zeta_N$ denotes an ordered permutation of  file indices $\{1,\dots,N  \}$ such that $\length_{\zeta_1}\leq\dots\leq\Omega_{\zeta_N}$, and ${\widetilde K}_i = K \sum_{j=i}^N q_{\zeta_j}$ denotes the expected number of receivers requesting a file with version length larger than $\Omega_i \tau$ bits. In \eqref{eq:gammaa}, $\mathcal F_\ell$ denotes a random set of $\ell$ files chosen from $\{\zeta_i,\dots, \zeta_N\}$ (with replacement) in an i.i.d manner according to  $\qbf$, and $ \Gamma_i(K,{\ell,n} ) $ denotes the probability that file $n \in \mathcal F_\ell$ requested by a set of  $\ell$   receives  maximizes  the quantity $\lambda(K,\ell,n)$. 
\end{theorem}
\begin{proof}
The proof is given in Appendix~\ref{app:Sym user}. %or The proof is given in \cite[Appendix C]{onlineVer}.
\end{proof}

%{\RED   The proof follows from deriving an upper bound on the per-demand coded multicast rate as in the proof of Theorem \ref{thm:demand} for a setting with symmetric receivers, and by taking its expected value over all possible demands. A detailed proof  is given in \cite[Appendix C]{onlineVer}.}

%Theorem~\ref{thm:Sym user} generalizes the results in \cite{ji15order}, and characterizes the expected coded multicast rate achieved in a network composed of symmetric receivers and non-symmetric files with unequal popularities $\qbf$ and lengths $\{\length_n \tau \}$.  %We note that given the symmetry assumptions in this section, Theorem~\ref{thm:Sym user} provides a tighter upper bound on the expected multicast rate compared to the one resulting from  specializing Theorem~\ref{thm:general} to a setting with symmetric receivers.

\section{The CC-CM Scheme Implemented with RF-GCC}\label{sec:RAP-GCCOptimization}

 In this section, we describe how  RF-GCC of  Sec.~\ref{sec:RAP} can be adopted by the CC-CM scheme of Sec.~\ref{sec: Multicast} to fill the receiver caches and to deliver the coded multicast portion of the transmissions in the delivery phase. The RF-GCC scheme is designed to be applicable to the scalable coding-based content delivery setting considered in this paper.  %(e.g., video streaming application or Gaussian sources) due to the inherent successively refinable nature of scalable encoding \cite{Cover}. 
 In fact, in line with Sec.~\ref{sec:RAP}, a version of a file used by RF-GCC is the combination of its base layer and a given number of its successive enhancement layers. %In line with Sec.~\ref{sec:RAP},  given two versions of a file, one is a degraded version of the other one if it is composed of fewer layers. 
 We use the rate upper bounds achieved with RF-GCC, provided in Sec.~\ref{sec:Achievable Rate}, to solve the optimization problem in \eqref{eq: Main Opt} and to characterize the rate-distortion-memory trade-off of the CC-CM scheme.

\subsection{Adopting  RF-GCC for  CC-CM: Scheme Description}\label{subsec: RAPP-GCC with SVC}
%In our proposed CC-CM scheme, in order to implement the caching phase and the portion of the delivery phase corresponding to coded multicast transmissions, we adopt the  RF-GCC scheme described in Sec.~\ref{sec:RAP}. 
%The RF-GCC scheme is applicable to the scalable coding-based content delivery setting considered in this paper  (e.g., video streaming application or Gaussian sources) due to the inherent successively refinable nature of scalable encoding \cite{Cover}. In fact, a version of a file is the combination of its base layer and a given number of its successive enhancement layers. In line with Sec.~\ref{sec:RAP},  given two versions of a file, one is a degraded version of the other one if it is composed of fewer layers. Specifically, the CC-CM scheme adopts the  RF-GCC for the scalable delivery of files as follows:

The CC-CM scheme adopts the  RF-GCC for the scalable delivery of files as follows:
\begin{itemize}
\item For a given set of cache allocations $\{M_{k,n} \}$ and per-receiver coded multicast rates $\{\widetilde R_{k,n} \} $, let $\Omega_{k,n} = M_{k,n} + \widetilde R_{k,n}$, $k=1,\dots,K$ and $n=1,\dots,N$, which we refer to as the {\em storing range} of receiver $k$ for file $n$. The storing range $\Omega_{k,n}$ is the rate with which file $n\in[N]$ is guaranteed to be delivered to receiver $k\in[K]$, upon request, through coded transmissions for any demand $\dbf$.  The number of source-samples $F$ and the storing range $\Omega_{k,n}$ (bits/sample) play the roles of parameter $\tau$ and $\Omega_{k,n}$ described in Sec.~\ref{sec:RAP}, respectively.
 
\item All  versions of the library files are partitioned into equal-length packets of $T$ bits. 
 
%\item[-] As in the  RF-GCC described in Sec.~\ref{sec:RAP}, the cache encoder in CC-CM is characterized by caching distributions $\pbf_k$,  $k=1,\dots,K$.   

\item During the caching phase, receiver $k\in[K]$ selects  
 $M_{k,n} F/T $ distinct packets uniformly at random from the $ \Omega_{k,n} F/T$ packets of version $k$ of file $n\in [N]$, where $\Omega_{k,n}$ determines the range of packets of file $n$ from which receiver $k$ is allowed to cache, hence the name storing range. Then, a packet from version $k$ of file $n$ is cached at receiver $k$ with probability
 \begin{align}
 p^c_{k,n} =   \frac{M_{k,n}}{M_{k,n}+ \widetilde R_{k,n}} ,\label{eq:P cache}
 \end{align} 
 which is in line with \eqref{eq: pc is prob} for $p_{k,n}\mu_k=M_{k,n}$ and $\Omega_{k,n} = M_{k,n}+\widetilde R_{k,n}$. The optimal values of $\{M_{k,n}\}$ and  $\{\widetilde R_{k,n}\}$  
are derived in terms of the rate budget $R$, cache sizes $\{ M_k\}$, and the demand distributions $\{\qbf_k\}$, by solving  \eqref{eq: Het Opt}, which we explain in Sec.~\ref{sec:RAP-GCCOptimization}. 

\item In the delivery phase, for a given demand  $\dbf\in\mathfrak D$, the  sender delivers the remaining $\widetilde R_{k,d_k}F/T$  missing (i.e., not cached) packets from the version requested by receiver receiver $k\in [K]$, via coded transmissions using the GCC scheme described in Sec.~\ref{sec:RAP}. 

\end{itemize}
Finally, the sender utilizes the remaining available rate from the total rate budget $R$ to transmit an additional layer, with rate $\widehat R_{k,{\bf d}}$, of file $d_k$ requested by receiver $k$ via uncoded transmissions. For a given demand $\dbf$, the per-receiver uncoded rates $\{\widehat R_{k,{\bf d}} \}$ can be determined based on a reverse water-filling approach similar to LC-U described in Sec.~\ref{sec: Unicast}. Since Gaussian sources are successively refinable, receiver $k $ is able to successfully recover file $d_k$ with rate $ \Omega_{k,d_k}+\widehat R_{k,{\bf d}}$.

Based on the results in Sec.~\ref{sec:RAP}, as $F\rightarrow\infty$, for any   $\dbf\in\mathfrak D$, the aggregate multicast rate achieved by    CC-CM  is upper bounded by $R^C \Big( \dbf,  \{\mu_k \}, \{\pbf_{k}\},\{ \Omega_{k,n} \}  \Big) +\sum\limits_{k=1}^K\widehat R_{k, \dbf}$, where $R^C\Big( \dbf,$ $ \{\mu_k \}, \{\pbf_{k}\}, \{ \length_{k,n} \}  \Big)$ is the aggregate coded rate achieved by  RF-GCC given in Theorem~\ref{thm:demand}. In Sec.~\ref{subsec:performance of cc-cm}, we use this upper bound to replace the first constraint of optimization  \eqref{eq: Main Opt}, as
\begin{align}
R^C \Big( \dbf,  \{\mu_k \}, \{\pbf_{k}\},\{ \Omega_{k,n} \}  \Big) +\sum\limits_{k=1}^K\widehat R_{k, \dbf} \leq R . \label{eq:constraint}
 \end{align}

\subsection{Discussion}
In this section, we briefly discuss some of the choices we made when designing the caching and delivery phases of  CC-CM  that adopts  RF-GCC. As explained   in  Sec.~\ref{sec: Multicast}, we partition the demand-dependent per-receiver rates $\{R_{k,\dbf}\}$ into two portions: a portion delivered through coded multicast that depends only on individual demands, $\{\widetilde R_{k,d_k}\}$, and another portion delivered through uncoded transmissions that depends on the entire demand, $\{\widehat R_{k,\dbf}\}$, which allows us to analytically evaluate the aggregate coded rate delivered by CC-CM.   
%which results in a more computationally tractable problem compared to the exponentially complex optimization problem given in \eqref{eq: Main Opt}. 
 Introducing demand-independent per-receiver rates $\{\widetilde R_{k,d_k}\}$ allows us to exploit coding opportunities during multicast transmissions while supporting a minimum reconstruction quality for each receiver request. This is achieved via defining the storing range. When adopting the RF caching strategy for  CC-CM, each receiver selects and caches various packets of a file version uniformly at random among a set of packets dictated by its storing range defined in \ref{subsec: RAPP-GCC with SVC}. This random population of the caches is a simple strategy to increase the distribution of distinct packets in the caches across the network, which is key for increasing the coding opportunities in the delivery phase  compared to traditional caching schemes that are based on local file popularity such as the Least Frequently Used (LFU) strategy\footnote{LFU is a local caching policy that, here, leads to all receiver caches having large overlaps,  limiting the coding opportunities.}\cite{ji15order}. Recall that in scalable encoding, an enhancement layer can not be used to improve the video quality without the base layer and all preceding enhancement layers. Hence, packets from a layer of a given file version can be potentially useless if all packets corresponding to its preceding layers are not received in their entirety. Using a caching strategy where receivers fill their caches starting from the lowest layer would limit the coding opportunities during the delivery phase, and result in a lower number of delivered enhancement layers.  However, with random caching only a subset of packets from different layers are available at a receiver. Therefore, due to scalable encoding all packets missing from these layers and preceding layers need to be delivered during the delivery phase in order to prevent packets that are cached  from being futile. To this end, we determine the minimum number of layers  
that we guarantee to fully deliver to each receiver based on the network setting, which  maps to the storing range, i.e., the lowest compression rate with which a file version can be delivered to that receiver, and utilize the remaining rate budget to deliver additional layers through uncoded transmissions by solving an optimization similar to LC-U.

\subsection{Rate-Distortion-Memory Trade-off with CC-CM}\label{subsec:performance of cc-cm}

In this section, our objective is to solve the optimization problem in \eqref {eq: Main Opt}. To this end, we %as discussed in Sec.~\ref{sec: Multicast}, we split the per-receiver rates $\{R_{k,\dbf} \}$ into two portions, $\{\widetilde R_{k,n} \}$ and $\{\widehat R_{k,\dbf} \}$, and 
adopt RF-GCC for the CC-CM scheme as in Sec.~\ref{subsec: RAPP-GCC with SVC}, which is equivalent to replacing the first constraint in   \eqref {eq: Main Opt} with  \eqref{eq:constraint}. Then, the optimal cache allocation $\{M_{k,n}^*\}$, per-receiver coded rates $\{\widetilde R_{k,n}^*\}$, and  per-receiver uncoded rates $\{\widehat R_{k, \dbf}^*\}$ are derived from %using the following optimization problem  
\begin{subequations}\label{eq: Het Opt}
	\begin{alignat}{4}
	& \text{min }
	&&  
	{\mathbb E}  \bigg[\frac{1}{K} \sum_{k=1}^{K}\sigma_{d_{k}}^{2}2^{-2(M_{k, d_{k}}+\widetilde R_{k, d_{k}}+ \widehat R_{k, \dbf})}\bigg]\label{eq: ObjectiveGeneral}\\
	& \text{s.t.}
	&& {  R}^C \Big( \dbf,  \{M_k \},  \{\pbf_{k}\},\{\length_{k,n} \}  \Big) +\sum\limits_{k=1}^K\widehat R_{k, \dbf} \leq R,   \hspace{0.5cm} \forall \dbf \in \mathfrak D \label{eq: Rate d} \\
	&&&  \length_{k,n} =  M_{k,n} +  \widetilde R_{k, n}   , \;\; M_{k,n} =   p_{k,n} M_{k}\hspace{1cm} \forall (k,n) \in [K]  \times [N]     \\
		%&&&  M_{k,n} =   p_{k,n} M_{k}   , \hspace{1.8cm} \forall (k,n) \in [K]  \times [N]     \\
	&&& \sum_{k=1}^{K}M_{k,n} \leq M_{k}, \hspace{1.7cm} \forall k \in [K]\label{eq: CacheGeneral}\\
	&&& M_{k,n}, \widetilde R_{k,n}, \widehat R_{k, \dbf}\geq 0,  \hspace{1cm}  \forall (k,n, \dbf) \in [K]  \times [N] \times \mathfrak D \label{eq: VariableGeneral}
	\end{alignat}
\end{subequations}

The optimization problem in \eqref{eq: Het Opt} is highly non-convex and has an exponential number of constraints due to \eqref{eq: Rate d}, which depends on the cardinality of $\mathfrak D$. We simplify the solution by relaxing \eqref{eq: Het Opt} and allowing the  rate constraint to be satisfied on average over all demands, and we replace \eqref{eq: Rate d} with
%Hence, we replace \eqref{eq: Rate d} with the following constraint
\begin{equation}\label{eq:constraint relaxed}
{ \bar R}^C \Big(\{\qbf_{k}\}, \{ M_k\}, \{\pbf_{k}\},\{ \length_{k,n} \}  \Big) + %\sum_{\dbf \in \mathfrak D}\Pi_{\dbf} 
{\mathbb E}\Big [\sum\limits_{k=1}^{K} \widehat R_{k, \dbf}\Big]\leq R,
\end{equation}
where ${ \bar R}^C \Big(\{\qbf_{k}\},\{ M_k\}, \{\pbf_{k}\},\{ \length_{k,n} \}  \Big)$ is given in Theorem \ref{thm:general}. 
%In the remainder of this section, 
In the following, we analyze the solution to the relaxed version of \eqref{eq: Het Opt} for settings with symmetry across receivers or files. % in the network.

\subsubsection{\bf Symmetry Across Receivers}~\label{sec:Symmetrical cc-cm}

%Consider a network with symmetric receivers where all receivers have equal-size caches and request files according to the same demand distribution, i.e. $M_{k} = M$, $q_{k,n} = q_{n}$,  for all $(k,n) \in [k]\times[N]$.  In this network, it is immediate to verify that the optimal caching distributions $\{\pbf_k \}$, and the corresponding cache allocations $\{ M_{k,n}\}$ are uniform across all the receivers, i.e., $p_{k,n} = p_n$ and $M_{k,n} = M_n$,  for all $(k,n) \in [k]\times[N]$.  Furthermore, all receivers have the same storing range for file $n\in[N]$, i.e., $\length_{k,n} = \length_{n}$    or equivalently, they have  the same per-receiver coded rate $\widetilde R_{k,n} =\widetilde R_{n}$. In this symmetric setting, constraint \eqref{eq:constraint relaxed} becomes ${\bar R}^{C}\Big( \qbf ,M ,\pbf, \{\length_n\} \Big )+   {\mathbb E}\Big [ \sum\limits_{k=1}^{K} \widehat R_{k, \dbf} \Big] \leq R$,  where ${\bar R}^{C}\Big( \qbf,M, \pbf, \{\length_n\} \Big )$,  the asymptotic expected  coded multicast rate achieved with the  RF-GCC scheme is in Theorem~\ref{thm:Sym user}. 
  
As described in Sec.~\ref{subsec:Symmetrical rec}, for symmetric receivers with equal-size caches and the same demand distribution, all receivers have the same storing range for file $n\in[N]$, i.e., $\length_{k,n} = \length_{n}$,  or equivalently, they have  the same per-receiver coded rate $\widetilde R_{k,n} =\widetilde R_{n}$. The asymptotic expected  coded multicast rate achieved with the  RF-GCC scheme in this setting is given in Theorem~\ref{thm:Sym user}. 

The performance of   CC-CM  depends on both the distortion-rate function of the sources according to which the files are generated and the file popularities. In order to see this dependency consider the following two cases. Consider a setting where files are generated in an i.i.d. fashion according to the same source distribution, and hence, they have the same distortion-rate function. In this case,   CC-CM    prioritizes the caching of more popular files. In order to simplify the analysis, as in \cite{ji15order}, we could use a caching distribution such that a set of the most popular files are cached with uniform probability, while all other less popular files are not cached at all. Using this caching policy for  RF-GCC  is proved in \cite{ji15order} to result in performance that is within a constant factor of the optimal one.    
 Alternatively, consider a setting where all files are equally popular but have different distortion-rate functions, which  corresponds to different variances $\{\sigma_n^2 \}$ for Gaussian sources. In this case,    CC-CM   prioritizes the caching of files that have higher distortion. Similarly to \cite{ji15order}, one could consider a simplified caching strategy, where a set of the files that are generated from sources with larger variance are cached with uniform probability, while all other files generated from sources with smaller variance are not cached at all.

 In line with \cite{ji15order}, we propose a simplified caching placement that takes into account both the popularity of the files and their distortion-rate functions. Let us divide the library files into two groups $\mathcal G_1$ and $\mathcal G_2$, with sizes $\widetilde N$ and $N-\widetilde N$, respectively, and assign fixed storing ranges $ \widetilde\Omega_1$ and $ \widetilde\Omega_2$ to all version of the files in groups $\mathcal G_1$ and $\mathcal G_2$, respectively. Then, the receivers fill their caches according to a truncated uniform caching distribution given as follows 
 \vspace{-0.7cm}
\begin{multicols}{2}
 \begin{equation}
p_n = \begin{cases}    {1}/ {\widetilde N} ,  &\hspace{0.5cm}  n \in  \mathcal G_1 \\ 0, \hspace{2cm}&\hspace{0.5cm}  n \in  \mathcal G_2 \end{cases},\notag
\end{equation}

 \begin{equation}
\Omega_n = \begin{cases} \widetilde\Omega_1,  &\hspace{0.5cm}  n \in  \mathcal G_1 \\ 
\widetilde\Omega_2 , \hspace{2cm}&\hspace{0.5cm}  n \in  \mathcal G_2 \end{cases},
\end{equation} 
\end{multicols}
where the cut-off index $\widetilde N \geq M$ and values $\widetilde\Omega_1$ and $\widetilde\Omega_2$ are a function of the system parameters, and are derived from solving   \eqref{eq: RLFU}. We refer to the resulting caching strategy as the Truncated Random Fractional (TRF) caching.  Intuitively, it is more likely that group $\mathcal G_1$ contains the more popular files that are also  generated from sources with higher variances. Group $\mathcal G_2$ contains all other files that are less popular and that are generated from sources with lower variances.  Then, 
\vspace{-1.8cm}
\begin{multicols}{2}
\begin{equation}\nonumber
M_{n} = \begin{cases} \widetilde M =M/ {\widetilde N} &\hspace{0.1cm}  n \in  \mathcal G_1  \\ 0 \hspace{1cm}&\hspace{0.2cm}  n \in  \mathcal G_2  \end{cases},
\end{equation}

\begin{equation}\nonumber
\widetilde R_{n}  = \begin{cases} \widetilde R_1 = \widetilde \Omega_1 -\widetilde M  &\hspace{0.2cm} n \in  \mathcal G_1 \\
 \widetilde R_2 =  \widetilde \Omega_2  \hspace{1cm}&\hspace{0.2cm}  n \in  \mathcal G_2  \end{cases}, 
\end{equation}
\end{multicols}
  and from \eqref{eq:P cache}, a packet of file $n\in[N]$ is cached at any receiver with probability 
\begin{equation}
p_n^c = \begin{cases} {\widetilde M}/ (\widetilde M+\widetilde R_1)  &\hspace{0.5cm}  n \in  \mathcal G_1 \\ 0 \hspace{2cm}&\hspace{0.5cm}  n \in  \mathcal G_2   \end{cases}.
\end{equation}

The optimal values for $\widetilde N$ (and hence $\widetilde M$), $\widetilde R_1$, $\widetilde R_2$ and $\{ \widehat R_{k,\dbf}\}$ are derived from % the following. % optimization problem.
\begin{subequations}\label{eq: RLFU}
	\begin{alignat}{4}
	&{\text{min}}
	\hspace{0.2cm} && \sum_{\dbf \in \mathfrak D} \Pi_{\dbf} \bigg(\frac{1}{K} \sum_{k=1}^{K}\sigma_{d_k}^{2}2^{-2(M_{d_k}+\widetilde R_{d_k} + \widehat R_{k, \dbf})} \bigg) \\
	& \text{s.t.}
	&& \min \bigg\{    
	\widetilde  \Psi^{(1)}\Big(  M,  \widetilde N,  \widetilde R_1,  \widetilde R_2, \widetilde G \Big),   \bar\Psi^{(2)}\Big( \qbf,\{ \length_{n}\}   \Big) \bigg\}  +  \sum_{\dbf \in \mathfrak D} \Pi_{\dbf}\sum\limits_{k=1}^{K} \widehat R_{k, \dbf} \leq
	R,\label{subeq:RateRLFU }\\
	&&& 	\widetilde  \Psi^{(1)}\Big(  M,  \widetilde N,  \widetilde R_1,  \widetilde R_2, \widetilde G  \Big) = \frac{\widetilde R_1 (  M + \widetilde N\widetilde R_1)}{ M}  \left (1-\left (\frac{\widetilde N\widetilde R_1}{  M + \widetilde N\widetilde R_1} \right)^{K\widetilde G } \right ) + K (1-\widetilde{G})\widetilde R_2 \nonumber \\
	%&&& \widetilde G  = \sum\limits_{n\in\mathcal G_1}   q_{n} \\
	&&& \widetilde N, \widetilde R_1,\widetilde R_2, \widehat R_{k, \dbf}\geq 0,  \hspace{0.3cm}  \forall (k, \dbf) \in [K]  \times \mathfrak D
	\end{alignat}
\end{subequations}
where  $\widetilde G  = \sum\limits_{n\in\mathcal G_1}   q_{n}$, and $\bar\Psi^{(2)}\Big( {\bar \Psi}^{(2)}\Big(\qbf,\mu,\pbf,\{{\length}_{n}\}    \Big)   \Big)$ is defined in \eqref{eq:psi 2 special case}.  The first term in \eqref{subeq:RateRLFU },  \linebreak $\widetilde  \Psi^{(1)}\Big(  M,  \widetilde N,  \widetilde R_1,  \widetilde R_2, \widetilde G \Big)$, is the expected coded  multicast rate achieved by  TRF-GCC for files in group $\mathcal G_1$, derived using Theorem~\ref{thm:Sym user} and by applying Jensen's inequality as explained in \cite[Appendix B]{ji15order}. The second term is the expected uncoded  multicast rate for files in group $\mathcal G_2$, from which no packet has been cached in the network. 
We refer to the resulting scheme as the  CC-CM scheme that adopts TRF-GCC.

\subsubsection{\bf  Symmetry Across Receivers and Files}~\label{sec: Uniform}

The simplest network setting consists of all receivers having equal-size caches, uniform demand distributions, and all files (sources) having the same distribution, i.e. 
$$M_{k} = M,\; q_{k,n} = \frac{1}{N}, \;\sigma_{n}^{2} = \sigma^{2} , \quad \text{for all } (k,n) \in [K]\times[N].$$
Due to the symmetry, it can be immediately verified that both the optimization problem in \eqref{eq: Main Opt} and the relaxed version of \eqref{eq: Het Opt} result in uniform caching distribution $p_{k,n} = \frac{1}{N}$, and a unique storing range $\Omega_{k,n} = \widetilde \Omega$ for all $(k,n)\in[K]\times[N]$. In this setting, $M_{k,n}=\widetilde M$ and $R_{k, \dbf}=  \widetilde R_{k, d_{k}}=\widetilde R$, and therefore from \eqref{eq:P cache} we have $p^c =  {\widetilde M}/{(\widetilde M+\widetilde R)}$. It is immediate to see that in this setting the optimal solution assigns  
$\widehat R_{k, \dbf}=0$ for any demand $\dbf\in\mathfrak D$, and that we only need to account for the per-receiver coded rates. The optimal values of $\widetilde M^*$ and $\widetilde R^*$ are derived using the relaxed version of problem \eqref{eq: Het Opt} by further particularizing the expected coded multicast rate given in Theorem~\ref{thm:Sym user} to symmetric files, as follows
\begin{equation}
\begin{aligned}
&{\text{min}}
& & \sigma^{2}2^{-2(\widetilde M+\widetilde R)} \\
& \text{s.t.}
& &  \Big(\widetilde M+\widetilde R\Big)\, \min\bigg\{  \frac{\widetilde R}{\widetilde M} \bigg (1- \Big (\frac{\widetilde R}{\widetilde M+\widetilde R}  \Big)^K \bigg)   ,\;
\bigg(1-\Big(1-\frac{1}{N}\Big)^K\bigg)^N \bigg\}
\leq R,\\
&&& \widetilde M \leq M, \;\;\; \widetilde M, \widetilde R \geq 0. 
%&&& \widetilde M, \widetilde R \geq 0.
\end{aligned}\label{eq: Uniform}
\end{equation}

% ********************************************************************************************************
% SIMULATIONS
% ********************************************************************************************************
 
\section{Numerical Results}~\label{sec:Simulations}
In this section, we numerically compare the performance of the LC-U and CC-CM content delivery schemes proposed in Secs.~\ref{sec: Unicast} and \ref{sec: Multicast} using the asymptotic closed-form results ($F\rightarrow\infty$) provided in  Secs.~\ref{sec:Symmetrical cc-cm} and \ref{sec: Uniform}. We consider a network composed of $K= 20$ receivers and a library with $N = 100$ files, which are requested by all receivers according to a Zipf distribution $\qbf$ with parameter $\alpha$, where $q_{n} = {n^{-\alpha}}/{\sum_{n= 1}^{N} n^{-\alpha}}$ for $n =1,\dots,N$.
%$$q_{n} = \frac{n^{-\alpha}}{\sum_{n= 1}^{N} n^{-\alpha}}   \hspace{2cm}   n =1,\dots,N.$$

Fig. \ref{fig:asymptotic} (a), displays the expected distortion achieved with the LC-U scheme {(exact)} and the CC-CM scheme {(upper bound)} using TRF-GCC. In order to reduce the complexity of problem \eqref{eq: RLFU}, we assume that the per-receiver uncoded rates $\{\widehat R_{k,\dbf}\}$ are independent of the demand and only depend on the file indices. Therefore, the CC-CM curve shown in Fig. \ref{fig:asymptotic} (a) provides an upper bound on the one resulting from solving  \eqref{eq: RLFU}. It is assumed that all receivers have the same cache size, $\alpha = 0.6$, and  $\sigma_{n}^{2}$ is uniformly distributed in the interval $[0.7,1.6]$. The distortions have been plotted (on a logarithmic scale) for rate budget values of $R \in \{2,5,8\}$  bits/sample as receiver cache sizes vary from $5$ to $100$ bits/sample.
As expected, CC-CM significantly outperforms LC-U  in terms of expected distortion. This means that for a given rate budget  $R$, CC-CM is able to deliver higher-rate file versions to the receivers, reducing their reconstruction distortions. Specifically, for rate budget $R= 2$ and cache size $M=50$, CC-CM achieves a $2.1\times$ reduction in expected distortion compared to LC-U, and for larger  rate budget $R=8$ the gain of CC-CM increases to $5.4$ for the same cache size $M=50$.

\begin{figure} %[h!]%\centering
	\begin{subfigure}{0.45\linewidth}\centering
		\includegraphics[width=0.87\linewidth]{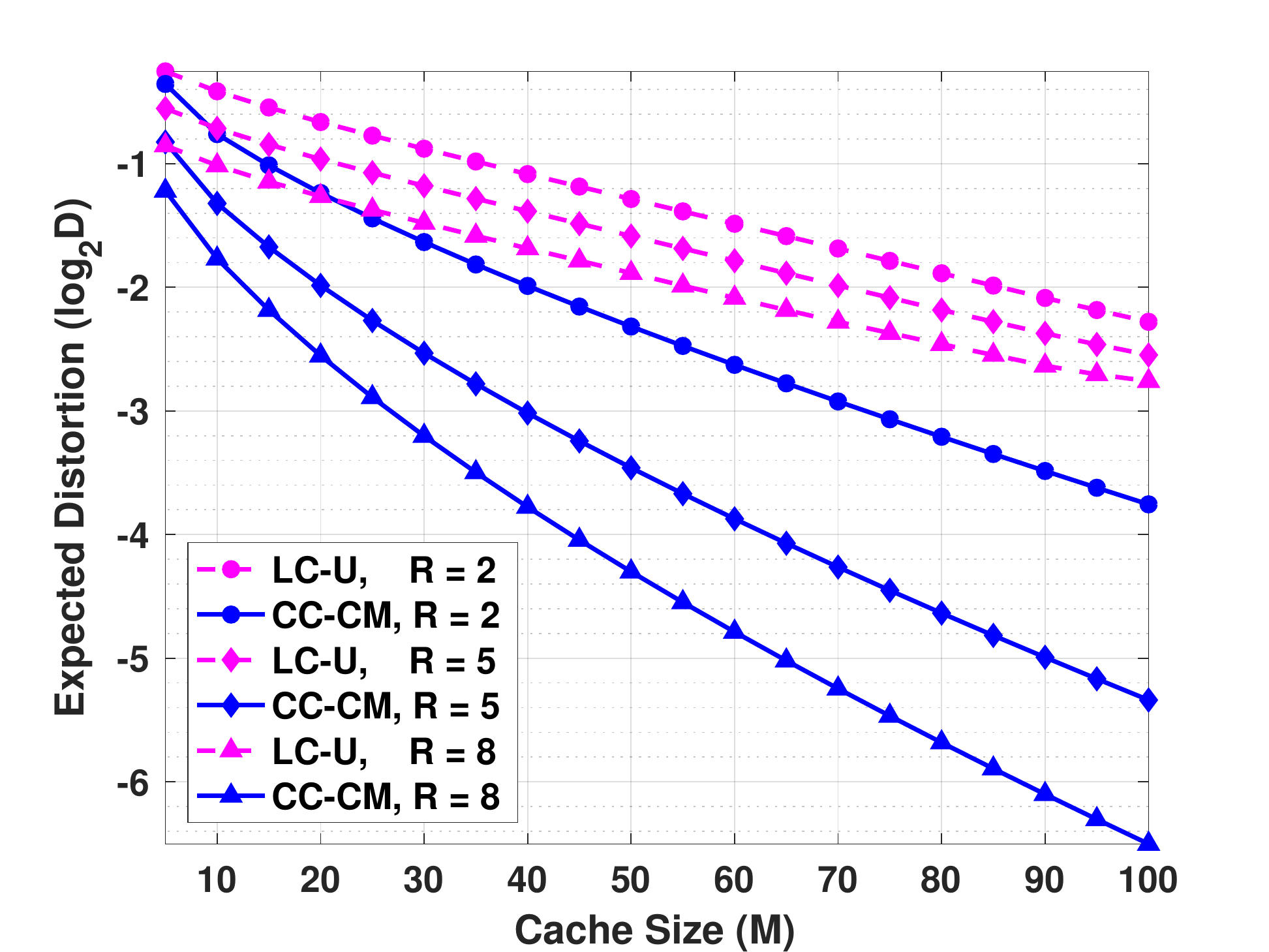}
		\subcaption{}
	\end{subfigure}\hspace*{\fill}
	\begin{subfigure}{0.45\linewidth}\centering
		\includegraphics[width=0.87\linewidth]{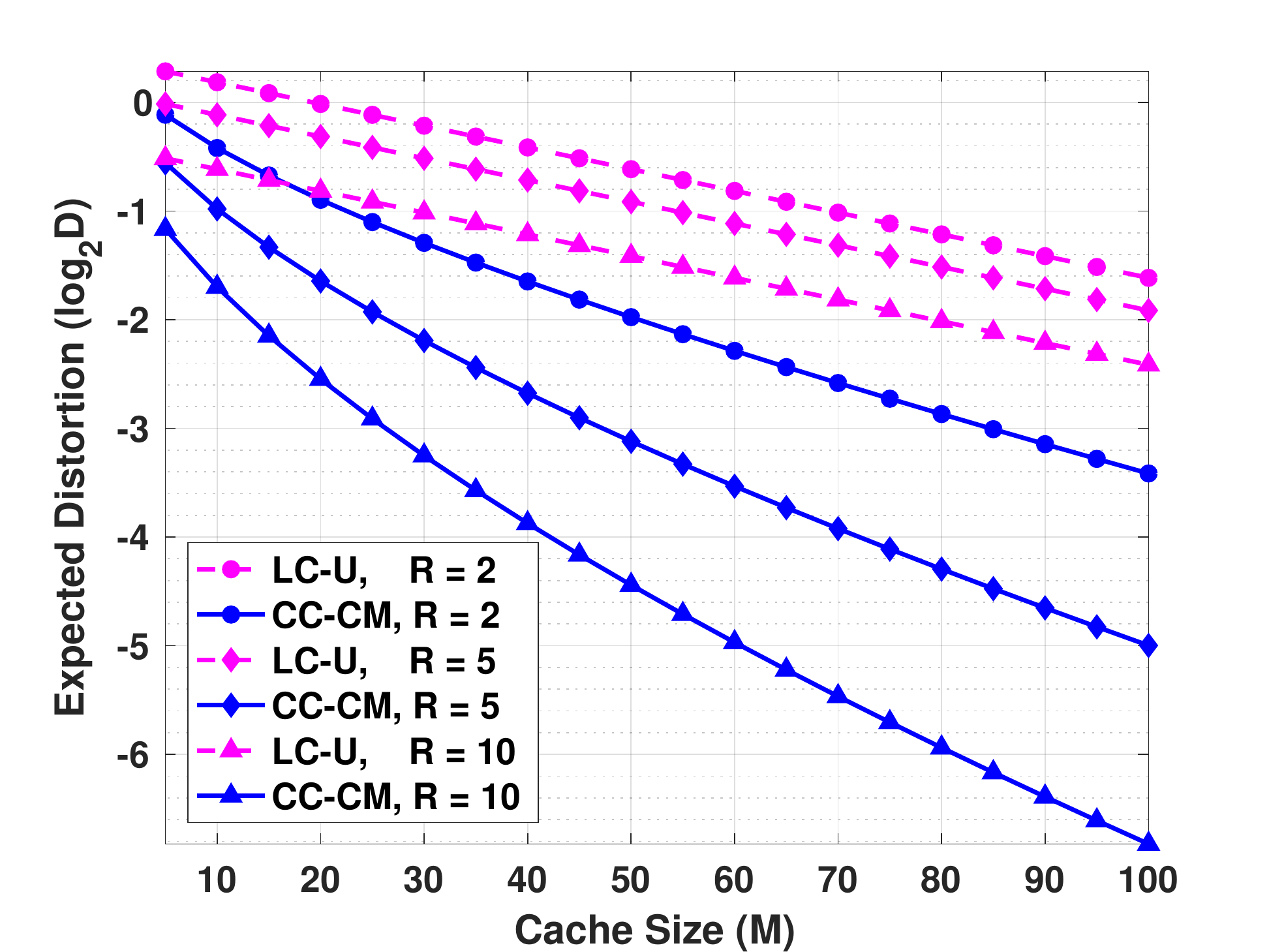}
				\subcaption{}
	\end{subfigure}\hspace*{\fill}
	\caption{Distortion-memory trade-off in a network with $K = 20$ receivers, $N= 100$ files, and Zipf demand distribution with parameter (a) $\alpha = 0.6$, and (b) $\alpha = 0$ (uniform demands).}
%, analyzed in Sec. \ref{sec:Symmetrical}.
\label{fig:asymptotic}
\end{figure}

In Fig. \ref{fig:asymptotic} (b), we consider a homogeneous network with uniform file popularity ($\alpha = 0$) and $\sigma_{n}^{2} = 1.5 $, for all $ n\in [N]$, $N = 100$. The expected distortions achieved for LC-U and CC-CM (using  RF-GCC) are plotted for the  rate budget values of $R \in \{2,5,10\}$ bits/sample as receiver cache sizes vary from $5$ to $100$ bits/sample. It is observed that the gains achieved by CC-CM are even higher in this scenario, which result from the increased coded multicast opportunities that arise when files have uniform popularity \cite{ji15order}. In this case, for $R=10$ and $M=50$, the expected distortion achieved with CC-CM is $9.5$ times less than with LC-U, and the improvement factor increase up to $14\times$  for cache capacity $M=70$.

% ********************************************************************************************************
% CONCLUSION
% ********************************************************************************************************

\section{Conclusion}~\label{sec: Conclusion}  
In this paper, we have investigated the use of caching in broadcast networks for enhancing video streaming quality, or in a more abstract sense, reducing source distortion. During low traffic hours, receivers cache low rate versions of the video files they are interested in, and during high traffic hours further enhancement layers are delivered to enhance the video playback quality. We have proposed two cache-aided content delivery schemes that differ in performance, computational complexity and required coding overhead. We have shown that while local caching and unicast transmission can be used to improve reconstruction distortion without the need of global coordination, the use of cooperative caching and coded multicast transmission is able to provide $10\times$ improvement in expected achievable distortion in a network with $20$ users and $100$ files by delivering more enhancement video layers with the same available broadcast resources. We have characterized the distortion-memory trade-offs for both schemes, and our numerical results have confirmed the gains that can be achieved by exploiting coding across the cached and requested content during multicast transmissions. As a subproblem to our main problem, we have generalized the setting in \cite{ji15order} to one that delivers different versions of library files to the users, thereby providing a solution to the lossy caching problem studied in \cite{yang2018coded}.

 %the setting considered in lossy cache-aided networks in the literature to networks with heterogeneity across files as well as users, and we have characterized a new upper bound on the rate-memory trade-off for a given set of preset user distortion requirements.

% ********************************************************************************************************
% Appendix
% ********************************************************************************************************

% ********************************************************************************************************
% Proofs
% ********************************************************************************************************
\begin{appendices}
	
\section{Proof of Theorem \ref{thm:demand}}\label{App:demand}
The proof is based on a generalization of the proof in \cite[Appendix A]{ji15order} to a setting where receivers have different cache sizes, different file preferences, and where they request degraded versions of the same file. 
We upper bound the asymptotic ($\factor \rightarrow\infty$) coded multicast rate achieved by the GCC algorithm.  %Let $R_{\dbf}^{GCC}$ denote the aggregate coded multicast rate required to deliver a version of size $\Omega_{k,d_k}F$ bits to receiver $k\in[N]$ for demand realization $\dbf$ using the GCC algorithm. 
 As described in \cite[Sec III-B]{ji15order}, the GCC algorithm applies two greedy graph coloring-based  algorithms, GCC$_1$ and GCC$_2$, to the index coding conflict graph $\mathcal H_{\Cbf,\Qbf}$, constructed based on the  packet-level cache configuration $\Cbf$ and demand realization $\Qbf$. Then, GCC determines the total number of distinct colors assigned by each algorithm to the graph vertices, and selects the  coloring that results in a smaller number of distinct colors. %We upper bound the rate achieved by  algorithms GCC$_1$ and GCC$_2$ following the same procedure as in \cite[Appendix A]{ji15order}.

%\subsection{Aggregate coded multicast rate for demand $\dbf$ achieved by  GCC$_1$:}
\subsection{Coded multicast rate achieved by  GCC$_1$ for demand $\dbf$:}
For a given vertex $v$ in the conflict graph $\mathcal H_{\Cbf,\Qbf}$, corresponding to packet $\alpha(v)$ requested by receiver $\beta(v)$, we refer to the unordered set of receivers $\{\beta(v), \eta(v)\}$ as the {\em receiver label} of $v$, which corresponds to the set of receivers either requesting or caching packet $\alpha(v)$. Note that by definition of the conflict graph, two vertices with the same receiver label are not connected via an edge, i.e., they do not interfere. Let $\mathcal J(\Cbf,\Qbf)$  denote the number of distinct colors assigned by algorithm GCC$_1$ to graph $\mathcal H_{\Cbf,\Qbf}$, which, by definition, is the number of independent sets\footnote{An independent set is a set of vertices in a graph, no two
of which are adjacent.} selected by the algorithm.  By construction,  GCC$_1$ generates independent sets that are composed of vertices with the same receiver label. %each of which is associated to a given receiver label.  
  We upper bound the number of independent sets in $\mathcal H_{\Cbf,\Qbf}$ by first splitting the graph into $K$ subgraphs, and upper bounding  $\mathcal J(\Cbf,\Qbf)$  with the sum of the number of independent sets found by GCC$_1$ in each of the $K$ subgraphs. For demand $\dbf$, let the ordered set $\order_1,\dots,\order_K$ denote a permutation of receiver indices $\{1,\dots,K \}$ such that $\length_{\order_1,d_{\order_1}}\leq\dots\leq\Omega_{\order_K, d_{\order_K}}$. Then, $\mathcal H_{\Cbf,\Qbf}$ is split into possibly $K$ subgraphs such that subgraph $i\in[K]$ is composed of all the vertices in $\Qbf$, denoted by $\Vc^{(i)}$, {that represent the requested packets that belong to the portion of files from bit $\length_{\order_{i-1},d_{\order_{i-1}}}   \factor$ to bit $\length_{\order_i,d_{\order_i}}\factor$ (Indexed from the beginning of a file) demanded by receivers $\{\order_i,\dots,\order_K\}$. Note that the first subgraph corresponding to $i=1$ is composed of all vertices that represent requested packets from the first $ \length_{\order_1,d_{\order_1}} \factor$ bits of all files in demand $\dbf$, and therefore we define $\length_{\order_0,d_{\order_0}} =0$. Subgraph $i$ is empty, i.e., has no vertices and edges, if $\length_{\order_{i-1},d_{\order_{i-1}}} = \length_{\order_i,d_{\order_i}}$, and consequently $\mathcal V^{(i)} = \emptyset$. By construction, subgraph $i$ only contains packets of files requested by receivers $\{\order_i,\dots,\order_K\}$, and after coloring graph $i$ there are no remaining packets requested by receiver $\order_i$ that need to be delivered.} Let us denote the number of independent sets in subgraph $i\in[K]$ by $\mathcal J_i(\Cbf,\Qbf)$.

We find an upper bound on $\mathcal J_i(\Cbf,\Qbf)$, $i\in[K]$, proceeding as in  \cite[Appendix A]{ji15order},  by enumerating all possible receiver labels, and by further upper bounding the number of independent sets that GCC$_1$ generates for each receiver label. {We define $\eta_i(v)$ as the set of all receivers in $\{\order_{i},\dots,\order_{K}\}$ which have cached packet $\alpha(v)$ corresponding to vertex $v\in\mathcal V^{(i)}$; therefore, $\eta_i(v)=\eta(v)\setminus\{\order_i,\dots,\order_K\}$. Let  $\mathcal K_{\ell} \subseteq \{\order_{i},\dots,\order_{K}\}$ denote a set of $\ell\in\{1,\dots, K-i+1  \}$ receivers,} and let $\mathcal J_{\Cbf,\Qbf}^{(i)}(\mathcal K_{\ell})$ denote the number of independent sets generated by GCC$_1$ with receiver label $\mathcal K_{\ell} $ for packet-level demand $\Qbf$ corresponding to subgraph $i$ with vertex set $\Vc^{(i)}$. As stated in \cite[Appendix A]{ji15order}, a necessary condition for the existence of an independent set with receiver label $\mathcal K_{\ell} = \{\beta(v), \eta_i(v)\} $ is that for any receiver $k\in \mathcal K_\ell$, there exist a vertex $v\in\Vc^{(i)}$ such that: 1) $\beta(v) = k$, i.e., receiver $k$ is requesting packet $\alpha(v)$, and 2) $\eta_i(v) = \mathcal K_{\ell}\setminus \{k\}$, i.e.,  $\alpha(v)$ is cached by all receivers in $\mathcal K_{\ell}\setminus \{k\}$, and not by any other receiver in $\{\order_i,\dots,\order_{K}\} $. Then, for a given $\Cbf$ and $\Qbf$, the number of generated independent sets  becomes
\begin{align}
\mathcal J_i(\Cbf,\Qbf) = \sum_{\ell=1}^{ K-i+1} \sum_{\mathcal K_{\ell} \subseteq \{\order_i,\dots,\order_{K}\}} \mathcal J_{\Cbf,\Qbf}^{(i)}(\mathcal K_{\ell}),\label{eq:indic}
\end{align}
%with  
\begin{flalign}
 \text{with} \hspace{4.5cm} \mathcal J_{\Cbf,\Qbf}^{(i)}(\mathcal K_{\ell}) &= \max\limits_{k\in\mathcal K_{\ell}}  \sum\limits_{ \substack{v \in \Vc^{(i)}:\\ \beta(v)=k} } \mathbbm{1}  \Big\{\eta_i(v) = \mathcal K_{\ell}\setminus \{k\}  \Big\},&\label{eq:indic2}
\end{flalign}
where %the indicator function  
$\mathbbm{1}  \Big\{\eta_i(v) = \mathcal K_{\ell}\setminus \{k\}\Big\}$ is a random variable, whose expected value gives the probability that vertex $v$ corresponding to file $d_k$ requested by receiver $k$  belongs to an independent set associated with receiver label $\mathcal K_\ell$. In other words, it indicates whether packet $\alpha(v)$ can be encoded into a linear codeword intended for all the receivers in $\mathcal K_\ell$.  For any vertex $v\in\mathcal V^{(i)}$, the indicator function $ Y_{\mathcal K_\ell,k} \triangleq \mathbbm{1}\Big\{\eta_i(v) = \mathcal K_{\ell}\setminus \{k\}\Big\}$ takes value 1 in the event that packet $\alpha(v)$ is cached at  all the receivers in $\mathcal K_\ell\setminus \{k\}$, and is not cached at any of the receivers in $\{\order_i,\dots,\order_{K}\}\setminus \mathcal K_\ell$. $Y_{\mathcal K_\ell,k}$ is a Bernoulli random variable with parameter
%Since  any receiver $u\in \mathcal [K]$ caches the packets from version $u$ of file $n\in[N]$ with probability $p_{u,n}$, this event occurs with probability $\phi (\mathcal K_\ell,k,d_k) $.
\begin{equation}
{\phi_i }(\mathcal K_\ell,k,d_k) \triangleq \prod\limits_{u\in \mathcal{K}_{\ell}\backslash \{k\}}  p^c_{u,d_{k}}  \prod\limits_{u\in { \{\order_i,\dots,\order_{K}\}}\setminus  \mathcal{K}_{\ell}}{ (1-p^c_{u,d_{k}})},\label{eq: phi app}
\end{equation}
where $p^c_{k,n}$ denotes the probability that a packet from version $k$ of file $n\in[N]$ is cached at receiver $k\in[K]$, and is given by
\begin{align}
p^c_{k,n} ={ \binom{ \Omega_{k,n} F/T-1 }{ M_{k,n} F/T-1 }  }\bigg/{ \binom {\Omega_{k,n} F/T}{ M_{k,n} F/T}   } =  \frac{M_{k,n}}{M_{k,n}+ \widetilde R_{k,n}}  = p_{k,n }\frac{M_{k}}{\Omega_{k,n}} .
\end{align}

Similar to \cite[Appendix A]{ji15order}, it can be is shown 
% in Lemma~\ref{lem:lemma1} in Appendix~\ref{App:lemma1} 
that as $\factor\rightarrow \infty$ with fixed $T$, we have
\begin{equation}
\lim\limits_{\factor/T \rightarrow\infty} \mathbb{P} \bigg(\bigg| \frac{Y_{\mathcal K_\ell,k} }{(\length_{\order_i,d_{\order_i}} -\length_{\order_{i-1},d_{\order_{i-1}}} ) \factor/T}  \, -\,  (1-p^c_{u,d_k})  \phi_i (\mathcal K_\ell,k,d_k) \bigg| \leq \epsilon
\bigg) = 1.
\end{equation}
% 
%
%\begin{align}
%\sum\limits_{ \substack{v \in \Vc^*:\\ \beta(v)=k}  } \mathbbm{1}  \Big\{\eta(v) = \mathcal K_{\ell}\setminus \{k\}\Big\}  =  (1-p_{k,d_{k}} )  \phi(\mathcal K_\ell,k,d_k) \length_{k^*,d_{k^*}} \frac{\factor}{T} + o(\factor)  .
%\end{align}
%Hence, as $\factor\rightarrow\infty$, the expected number of independent sets generated with receiver label $\mathcal K_\ell$ is upper bounded by
As $\factor\rightarrow\infty$, the expected number of independent sets %generated 
with label $\mathcal K_\ell$ is upper bounded by
\begin{align}
 {\mathbb E}_{\Cbf}\bigg[   \mathcal J_{\Cbf,\Qbf}^{(i)}(\mathcal K_{\ell})     \Big|\Cbf  \bigg] &=  {\mathbb E}_{\Cbf}\bigg[          \max\limits_{k\in\mathcal K_{\ell}}  \sum\limits_{ \substack{v \in \Vc^{(i)}:\\ \beta(v)=k} } \mathbbm{1}  \Big\{\eta_i(v) = \mathcal K_{\ell}\setminus \{k\}\Big\}                    \Big|\Cbf  \bigg] \notag\\
& \leq   \max\limits_{k\in\mathcal K_{\ell}}   \Big\{
(1-p^c_{k,d_{k}})   \phi_i(\mathcal K_\ell,k,d_k)   \Big\}  (\length_{\order_i,d_{\order_i}} -\length_{\order_{i-1},d_{\order_{i-1}}} )  \frac{\factor}{T} .
\label{eq:J1}
\end{align}

Therefore, an upper bound on the asymptotic coded multicast rate for demand $\dbf$ and a given set of caching distributions $\{\pbf_k \}_{k=1}^K$, can be derived from \eqref{eq:indic} and \eqref{eq:J1} %and \eqref{eq:J2} 
as follows
\begin{align}
\Psi_{\dbf}^{(1)}\Big(  \{\mu_k\}  , \{\pbf_{k}\},& \{\length_{k,n}  \}\Big)
\triangleq
\frac{1}{\factor/T}\,{\mathbb E}_{\Cbf}  \sum_{i=1}^K \Big[   \mathcal J_i(\Cbf,\Qbf)    \Big|\Cbf \Big]\notag\\
 &\leq 
 \sum_{i=1}^K \sum_{\ell=1}^{K-i+1} \sum_{\mathcal K_{\ell} \subseteq\{\order_i,\dots,\order_{K}\}} \; (\length_{\order_i,d_{\order_i}} -\length_{\order_{i-1},d_{\order_{i-1}}} )  \max\limits_{k\in\mathcal K_{\ell}} \; 
(1-p^c_{k,d_{k}} )   \phi_i (\mathcal K_\ell,k,d_k)    \notag .
\end{align}

\subsection{Coded multicast rate achieved by  GCC$_2$ for demand $\dbf$:}
Algorithm GCC$_2$ corresponds to uncoded multicast transmissions. As described in \cite[Sec III-B]{ji15order}, GCC$_2$ randomly selects a vertex $v$ in the conflict graph $\mathcal H_{\Cbf,\Qbf}$ and  generates independent sets composed of all vertices representing the same packet $\alpha(v)$ represented by vertex $v$. Then, it assigns the same color to all the vertices in each independent set. This corresponds to transmitting a total number of packets equal to the number of distinct requested packets. In order to evaluate this value for a given set of cache sizes $\{\mu_k\}_{k=1}^K$, we upper bound it with the number of packets that need to be delivered in a scheme where receiver $k\in[K]$ has cached the first $p_{k,d_k}\mu_k\,{\tau}/{T}$ packets from the total $\length_{k,d_k}\,{\tau}/{T}$ packets  of version  $k$ of file $d_k\in[N]$. In this case, for a requested file $n\ni\dbf$ the longest requested version of file $n$, i.e., $\argmax\limits_{k:d_k = n} \, \length_{k,d_k}$, needs to be transmitted. Given that receivers have heterogeneous cache sizes and caching distributions, only $\argmin\limits_{k:d_k = n} \, p_{k,d_k}\mu_k \, {\tau}/{T}$ packets of this file have been cached by all receivers requesting this file, and therefore, the multicast rate for demand $\dbf$ is upper bounded by
\begin{align}
\Psi_{\dbf}^{(2)}\Big(  \{\mu_k\}  , \{\pbf_{k}\}, \{\length_{k,n}  \} \Big)  \triangleq  \sum_{n =1}^N  \mathbbm{1} \{n \ni \dbf\} \Big(\max\limits_{k:d_k = n}    \,\Omega_{k,n}  - \min\limits_{k:d_k = n}    \,p_{k,n}\mu_k    \Big)
 . \label{eq:GCC2 app1}
\end{align}

% ********************************************************************************************************
% Proof Theorem 2
% ********************************************************************************************************

\section{Proof of Theorem \ref{thm:general}}\label{App:general}
We derive an upper bound on the  expected coded multicast rate required to deliver a version of file $n\in[N]$ with rate $\length_{k,n}$ to receiver $k\in[K]$, 
by  taking the expected value of the rate given Theorem~\ref{thm:demand}, and derived in Appendix \ref {App:demand}, over all possible demands $\dbf\in\mathfrak D$.

\subsection{Expected coded multicast rate achieved by  GCC$_1$}
Let $\lambda_i(\mathcal K_\ell,k,n ) \eqdef (1-p^c_{k,n}) \phi_i(\mathcal K_\ell,k,n ) $ for $\phi_i(\mathcal K_\ell,k,n ) $ defined in \eqref{eq: phi app}, then by taking the expectation of the rate in Theorem~\ref{thm:demand}, we have
\begin{align}
\mathbb E \Big[ \Psi_{\dbf}^{(1)}\Big(  \{\mu_k \},  &\{\pbf_{k} \}, \{\length_{k,n}\} \Big)   \Big]   
=\mathbb E \bigg[ \sum_{i=1}^{K} \sum_{\ell=1}^{K-i+1} \sum_{\mathcal K_{\ell} \subseteq \{\order_{i},\dots,\order_K  \}} \Big(\length_{\order_i,d_{\order_i}} -  \length_{\order_{i-1},d_{\order_{i-1}}}  \Big) \max\limits_{k \in \mathcal K_\ell}   \lambda_i(\mathcal K_\ell,k,d_k)    \bigg] 
\notag\\%\label{eq:avg main} \\
& = \sum_{i=1}^{K} \sum_{\ell=1}^{K-i+1} \mathbb E \bigg[\sum_{\mathcal K_{\ell} \subseteq \{\order_{i},\dots,\order_K  \}} \Big(\length_{\order_i,d_{\order_i}} -  \length_{\order_{i-1},d_{\order_{i-1}}}  \Big) \max\limits_{k \in \mathcal K_\ell}   \lambda_i(\mathcal K_\ell,k,d_k)    \bigg], \label{eq: avg first step}
\end{align}
%The right hand side term in \eqref{eq:avg main} can be upper bounded as follows:
%In \eqref{eq: avg first step}, 
where the expectation is taken over all subsets $\mathcal K_\ell$ of the set $\{\order_i,\dots,\order_K \}$ which is a function of the random demand realization $\dbf$, and therefore, {the order of the expectation and summation can not be exchanged}. Consequently, we upper bound  \eqref{eq: avg first step} with the delivery rate in a network where  receiver $k\in[K]$ requests equal-length versions of all file in the library each of size $ \length^*_{k} \tau$ bits with $\length^*_{k} \triangleq\max_{n\in[N]} \length_{k,n} $, i.e., it requests versions of files with the largest rate. For a given set of $\length^*_{1}, \dots, \length^*_{K}$, let the ordered set $\order_1^*,\dots,\order_K^*$ denote a permutation of receiver indices $\{1,\dots,K \}$ such that $\length_{\order_1^*}^*\leq\dots\leq\Omega_{\order_K^*}^*$. Note that the set $\order_1^*,\dots,\order_K^*$ is independent of the random demand $\dbf$.    Then, from \eqref{eq: avg first step}  we have
\begin{align}
\mathbb E \Big[ &\Psi_{\dbf}^{(1)}\Big(  \{\mu_k \},  \{\pbf_{k} \}, \{\length_{k,n}\} \Big)   \Big]  \leq 
\sum_{i=1}^{K} \sum_{\ell=1}^{K-i+1} \mathbb E \bigg[\sum_{\mathcal K_{\ell} \subseteq \{\order_{i}^*,\dots,\order_K^*  \}} \Big(\length_{\order_i^*}^* -  \length_{\order_{i-1}^*}^*  \Big) \max\limits_{k \in \mathcal K_\ell}   \lambda_i(\mathcal K_\ell,k,d_k)    \bigg] \notag\\
&\stackrel{(a)}{=}   \sum_{i=1}^{K} \sum_{\ell=1}^{K-i+1}  \sum_{ \dbf \in {\mathfrak D}}  \Big(\prod_{k \in [K]} q_{k,d_k}\Big)\sum_{\mathcal K_{\ell} \subseteq \{\order_{i}^*,\dots,\order_K^*  \}} \Big(\length_{\order_i^*}^* -  \length_{\order_{i-1}^*}^*  \Big) \max\limits_{k \in \mathcal K_\ell}   \lambda_i(\mathcal K_\ell,k,d_k)      \notag\\
&\stackrel{(b)}{=}   \sum_{i=1}^{K} \sum_{\ell=1}^{K-i+1}  \sum_{ \dbf \in {\mathfrak D}}  \Big(\prod_{k \in { [K]} } q_{k,d_k}\Big)\sum_{\mathcal K_{\ell} \subseteq \{\order_{i}^*,\dots,\order_K^*  \}} \Big(\length_{\order_i^*}^* -  \length_{\order_{i-1}^*}^*  \Big)  \notag\\
&\qquad \qquad\bigg(   \sum_{n =1}^{N}   \sum_{k \in {\mathcal K}_\ell  } \mathbbm{1} \Big\{ (k,n)= \argmax\limits_{( s,t): s\in\mathcal K_\ell, t = d_s  }\,   \,\lambda_i(\mathcal K_\ell,s,t)   \Big\} \;.\;  \lambda_i(\mathcal K_\ell,k,  n)    \bigg)   \notag\\
%&=\sum_{i=1}^{K} \sum_{\ell=1}^{K-i+1}  \sum_{\mathcal K_{\ell} \subseteq \{\order_{i}^*,\dots,\order_K^*  \}} \sum_{n =1}^{N}   \sum_{k \in {\mathcal K}_\ell  } \Big(\length_{\order_i^*}^* -  \length_{\order_{i-1}^*}^*  \Big) \lambda_i(\mathcal K_\ell,k,  n)  \notag\\
%&\qquad \qquad \sum_{ \fbf \in {\mathcal F} ({  \mathcal K_\ell}) }  \Big(\prod_{k \in {  \mathcal K_\ell}} q_{k,f_k}\Big)  \mathbbm{1} \Big\{ (k,n)= \argmax\limits_{( s,t): s\in\mathcal K_\ell, t = f_s  }\,   \,\lambda_i(\mathcal K_\ell,s,t)   \Big\}    \notag\\
&\stackrel{(c)}{=}\sum_{i=1}^{K} \sum_{\ell=1}^{K-i+1}  \sum_{\mathcal K_{\ell} \subseteq \{\order_{i}^*,\dots,\order_K^*  \}} \sum_{n =1}^{N}   \sum_{k \in {\mathcal K}_\ell  }
\Big(\length_{\order_i^*}^* -  \length_{\order_{i-1}^*}^*  \Big) \lambda_i(\mathcal K_\ell,k,  n)  \, \mathbb E\Big[ 
\mathbbm{1} \Big\{ (k,n)= \argmax\limits_{( s,t): s\in\mathcal K_\ell, t = f_s  }\,   \,\lambda_i(\mathcal K_\ell,s,t)   \Big\}  \Big]  \notag\\
%&=\sum_{i=1}^{K} \sum_{\ell=1}^{K-i+1}  \sum_{\mathcal K_{\ell} \subseteq \{\order_{i}^*,\dots,\order_K^*  \}} \sum_{n =1}^{N}   \sum_{k \in {\mathcal K}_\ell  }   \Big(\length_{\order_i^*}^* -  \length_{\order_{i-1}^*}^*  \Big) \lambda_i(\mathcal K_\ell,k,  n)  \,\mathbb P\Big(  (k,n)= \argmax\limits_{( s,t): s\in\mathcal K_\ell, t = f_s }\,   \,\lambda_i(\mathcal K_\ell,s,t)  \Big)  \notag\\
&\stackrel{(d)}{=}\sum_{i=1}^{K} \sum_{\ell=1}^{K-i+1}  \sum_{\mathcal K_{\ell} \subseteq \{\order_{i}^*,\dots,\order_K^*  \}} \sum_{n =1}^{N}   \sum_{k \in {\mathcal K}_\ell  }
\Big(\length_{\order_i^*}^* -  \length_{\order_{i-1}^*}^*  \Big) \lambda_i(\mathcal K_\ell,k,  n)  \,\Gamma_i(\mathcal K_\ell,k,n) ,\label{eq:exp}
\end{align}
where $(a)$ follows by writing the expectation with
respect to the demand vector $\dbf\in{\mathfrak D}$. %and ${\mathcal F}(\mathcal K)$ denotes the  set of all possible demand vectors made by the receivers in set $\mathcal K$. Then, 
Then $(b)$ follows from
replacing $\max_{k\in\mathcal K_\ell}\lambda_i(\mathcal K_\ell,k,  d_k)$ with a sum over all possible file-receiver indices $(k,n)$ of $\lambda_i(\mathcal K_\ell,k,  n)$ multiplied by the indicator function that picks the maximum value, and $(c)$ follows since only the indicator function depends on the demand, and $(d)$ follows by denoting
\begin{align}
\Gamma_i({\mathcal K_\ell,k,n}) \triangleq \mathbb P\Big( (k,n)= \argmax\limits_{ ( s,t): s\in\mathcal K_\ell, t = f_s }\,\lambda_i(\mathcal K_\ell,s,t)     \Big),
\end{align}
which is the probability that file $n\ni \fbf$ requested by receiver $k\in\mathcal K_\ell$ maximizes  the quantity $\lambda_i(\mathcal K_\ell,s,t) $, and where $\sum_{n=1}^{N}\sum_{k \in \mathcal K_{\ell}}   \Gamma_i({\mathcal K_\ell,k,n})  =1$. 
Therefore, the expected coded multicast rate achieved by GGC$_1$ for a given set of caching distributions $\{\pbf_k  \}_{k=1}^K$ is upper bounded by

\scalebox{0.9}{\parbox{\linewidth}{%
		\begin{align}
		{\bar \Psi}^{(1)}\Big( \{\qbf_{k} \},  \{\mu_k  \} ,\{\pbf_{k} \}, \{\length_{k,n}\}\Big) & \triangleq 
		\sum_{i=1}^{K} \sum_{\ell=1}^{K-i+1}  \sum_{n =1}^{N}\sum_{\mathcal K_{\ell} \subseteq \{\order_{i}^*,\dots,\order_K^*  \}}    \sum_{k \in {\mathcal K}_\ell  }
		\Big(\length_{\order_i^*}^* -  \length_{\order_{i-1}^*}^*  \Big) \lambda_i(\mathcal K_\ell,k,  n)  \,\Gamma_i(\mathcal K_\ell,k,n).
		\notag
		\end{align}
}}

\subsection{Expected coded multicast rate achieved by  GCC$_2$}
Taking the expectation of the rate given in \eqref{eq:GCC2 app1} over all demand realizations $\dbf\in\mathfrak D$ results in
\begin{align}
{\bar \Psi}^{(2)} \Big(  \{\qbf_{k} \}, \{\Omega_{k,n}\}  \Big)\triangleq
\mathbb E \Big[ \Psi_{\dbf}^{(2)}\Big( \{\Omega_{k,n}\}\Big)   \Big]  
%&=\mathbb E \Big[\sum_{n =1}^N  \mathbbm{1} \{n \ni \dbf\} \max\limits_{k:d_k = n}  \;\Omega_{k,n}   \Big] \notag\\
&=\sum_{n=1}^N  \mathbb E\Big[\mathbbm{1}\{n\ni \dbf\} \Big(\max\limits_{k:d_k=n}  \Omega_{k,n}  -  \min\limits_{k:d_k=n}  p_{k,n}\mu_k \Big) \Big] \notag\\
&\stackrel{(a)}{\leq}    \sum_{n=1}^N \mathbb P\Big( \mathbbm{1} \{n \ni \dbf\}\Big)  \Big(\max\limits_{k\in[K]}   \length_{k,n}  -  \min\limits_{k\in[K]}  p_{k,n}\mu_k \Big) \notag\\
%&= \sum_{n=1}^N \max\limits_{k} \{\omega_{k,n}b\} \Big(1-\mathbb P(1\{n \notin \dbf\})\Big) \notag\\
%&= \sum_{n=1}^N \max\limits_{k} \{\omega_{k,n}b\} \Big(1-\prod_{k=1}^K (1-q_{k,n})\Big) \notag \\
&= \sum_{n=1}^N \Big(1-\prod_{k=1}^K (1-q_{k,n})\Big)\Big(\max\limits_{k\in[K]}   \length_{k,n}  -  \min\limits_{k\in[K]}  p_{k,n}\mu_k \Big)  \notag ,
\end{align}
where   $(a)$ follows since  $\max\limits_{k:d_k=n}  \Omega_{k,n} \leq \max\limits_{k\in[K]}   \Omega_{k,n}$ and  $\min\limits_{k:d_k=n}  p_{k,n}\mu_k \geq \min\limits_{k\in[K]}   p_{k,n}\mu_k$ for any  $\dbf\in\mathfrak D$.
%

% ********************************************************************************************************
% Proof Theorem 3
% ********************************************************************************************************
\section{Proof of Theorem \ref{thm:Sym user}}\label{app:Sym user}
The proof follows steps similar to those in \cite[Appendix A]{ji15order}, and based on the explanations given in Appendix \ref{App:demand}. We upper bound the  number of independent sets in ${\mathcal H}_{{\bf C}, {\bf Q}}$, by splitting the graph into $N$ subgraphs such that subgraph $i$ contains a subset of the packets of all requested files that have version length equal or larger than the $i^\text{th}$ shortest version length. Let $J_i(\Cbf,\Qbf)$  denote the number of independent sets   found by Algorithm GCC$_1$ in subgraph $i$. For a given $\Cbf$ and $\Qbf$,  we upper bound the delivery rate with  $\sum_{i}J_i(\Cbf,\Qbf)$. Let the ordered set $\zeta_1,\dots,\zeta_N$ denote a permutation of the file indices $\{1,\dots,N \}$ such that $\Omega_{\zeta_1}\leq\dots\leq\Omega_{\zeta_N}$. Then, for a given demand, subgraph $i$ is composed of all (if any)  vertices in $\mathcal V$,  denoted by $\mathcal V^{(i)}$, corresponding to packets in $\bf Q$ that belong to the portion of requested files from bit $\Omega_{\zeta_{i-1}}\tau$ to bit $\Omega_{\zeta_i}\tau$. %Therefore, for each demand $\dbf$, subgraph $i$ only contains packets demanded by receivers requesting a file with version length lager than $\Omega_{\zeta_i}\tau$ bits.   
Let us denote the set of receivers requesting a packet in subgraph $i$ by ${\mathcal K}^{(i)}$. Following the procedure in Appendix A, the normalized number of independent sets generated by the algorithm becomes
\begin{align}
\frac{1}{\tau/T} \mathbb E_{\bf C}\Big[ \sum_{i=1}^N\mathcal J_i(\Cbf,\Qbf) \Big]&= \sum_{i=1}^N \sum_{\ell=1}^{|{\mathcal K}^{(i)}|}\sum_{{\mathcal K}_{\ell} \in {\mathcal K}^{(i)}} \Big( \Omega_{\zeta_{i}} - \Omega_{\zeta_{i-1}}  \Big) \max\limits_{n\in{\bf f}_{\ell}}  \sum\limits_{ \substack{v \in \Vc^{(i)}:\\ \alpha(v)\text{ belongs to } n} } \mathbbm{1}  \Big\{\eta_i(v) = \mathcal K_{\ell}\setminus \{k\}  \Big\}\notag\\
&= \sum_{i=1}^N \sum_{\ell=1}^{ |{\mathcal K}^{(i)}|} \binom{|{\mathcal K}^{(i)}|}{\ell}  \Big( \Omega_{\zeta_{i}} - \Omega_{\zeta_{i-1}}  \Big) \max\limits_{n\in{\bf f}_{\ell}}  \lambda(|{\mathcal K}^{(i)}|,\ell, n) , \notag
\end{align}
which follows due to the homogeneity across receivers, with $\lambda(K,\ell, n) $ given as
\begin{align}
\lambda(K,\ell, n) =  (   p_n^c  )^{\ell-1} ( 1-p_n^c   )^{   K-\ell+1} .
\end{align} 
The expected delivery rate can be upper bounded by taking the expectation over all demands as:
\begin{align}
{\bar \Psi}^{(1)}\Big(\qbf,\mu,\pbf,\{{\length}_{n}\}    \Big)   &\leq \mathbb E\bigg[ \sum_{i=1}^N \sum_{\ell=1}^{ |{\mathcal K}^{(i)}|} \binom{|{\mathcal K}^{(i)}|}{\ell}  \Big( \Omega_{\zeta_{i}} - \Omega_{\zeta_{i-1}}  \Big) \max\limits_{n\in{\bf f}_{\ell}}   \lambda(|{\mathcal K}^{(i)}|,\ell, n) \bigg]\notag\\
& =  \sum_{i=1}^N\Big( \Omega_{\zeta_{i}} - \Omega_{\zeta_{i-1}}  \Big) \mathbb E\bigg[ \sum_{\ell=1}^{ |{\mathcal K}^{(i)}|} \binom{|{\mathcal K}^{(i)}|}{\ell}   \max\limits_{n\in{\bf f}_{\ell}}  \lambda(|{\mathcal K}^{(i)}|,\ell, n) \bigg]\notag\\
& \stackrel{(a)}{\leq}  \sum_{i=1}^N\Big( \Omega_{\zeta_{i}} - \Omega_{\zeta_{i-1}}  \Big) \mathbb E\bigg[ \sum_{\ell=1}^{ |{\mathcal K}^{(i)}|} \binom{|{\mathcal K}^{(i)}|}{\ell}   \sum_{  n\in \{\zeta_i,\dots, \zeta_N \}}\Gamma_i(|{\mathcal K}^{(i)}|,\ell, n) \lambda(|{\mathcal K}^{(i)}|,\ell, n) \bigg]\notag\\
& \stackrel{(b)}{\leq}  \sum_{i=1}^N\Big( \Omega_{\zeta_{i}} - \Omega_{\zeta_{i-1}}  \Big)  \sum_{\ell=1}^{ {\widetilde K}_i } \binom{{\widetilde K}_i }{\ell}   \Gamma_i({\widetilde K}_i , \ell, n)   \lambda({\widetilde K}_i,\ell, n)  \notag,
\end{align}
where $(a)$ follows using the same trick as in Appendix B with 
\begin{align}
\Gamma_i(K,{\ell,n} ) = \mathbb P\Big( n= \argmax\limits_{t\in\mathcal F_\ell} \;\; (   p_t^c  )^{\ell-1} ( 1-p_t^c   )^{K-\ell+1}    \Big),
\end{align}
denoting the probability that file $n\in{\mathcal F}_\ell$ chosen from a random set of $\ell$ files in $\{\zeta_i,\dots, \zeta_N \}$ maximizes $\lambda(K,{\ell,n} )$,   and $(b)$  follows from Jensen's inequality due to the concavity of the function over which the expectation is taken.  $  {\widetilde K}_i = K \sum_{j=i}^N q_{\zeta_j}$ denoting the expected number of receivers in set ${\mathcal K}^{(i)}$.

\end{appendices}

% ********************************************************************************************************
% BIBLIOGRAPHY
% ********************************************************************************************************

\bibliographystyle{IEEEtran}
\bibliography{References2}

\end{document}